\documentclass[preliminary,copyright,creativecommons]{eptcs}
\providecommand{\event}{MeCBIC2010}

\providecommand{\firstpage}{1}

\usepackage{graphicx}
\usepackage{amsmath}
\usepackage{amssymb}
\usepackage{amsthm}
\usepackage{needspace}
\usepackage{enumerate}
\usepackage{color}
\usepackage[noend]{algorithmic}
\usepackage{listings}

\lstdefinelanguage{Algorithm}
{
 morecomment = [l]{//}, 
 morecomment = [l]{///},
 morecomment = [s]{/*}{*/},
 morestring=[b]", 
 sensitive = false,
 morekeywords = {
   Input, Output, loop, if, then, else, break, continue, end, null
	}
}

\newtheorem{theorem}{Theorem}

\newcommand{\myway}[1]{\raisebox{-4pt}{\rule{0pt}{16pt}}\colorbox[rgb]{.7,.7,.7}{#1}}

\newenvironment{definition}[1][]
   {\refstepcounter{theorem} \par\medskip\noindent
   {\bf Definition~\thetheorem \ifx #1 \else ~(#1)\fi.} 
\ignorespaces }
   {\par\medskip }

\theoremstyle{remark}
\newtheorem{remark}[theorem]{Remark}
\newtheorem{algorithm}[theorem]{\textbf{Algorithm}}
\newtheorem{fact}[theorem]{Fact}

\newcommand{\modmin}{\mathtt{min}}

\newcommand{\modmax}{\mathtt{max}}
\newcommand{\modone}{\mathtt{one}}
\newcommand{\modspread}{\mathtt{spread}}
\newcommand{\modrepl}{\mathtt{repl}}

\title{Edge- and Node-Disjoint Paths in P~Systems}
\author{Michael~J.~Dinneen, Yun-Bum~Kim, and Radu~Nicolescu
\institute{Department of Computer Science, University of Auckland,\\
Private Bag 92019, Auckland, New Zealand}
\email{\{mjd,yun,radu\}@cs.auckland.ac.nz}
}
\def\titlerunning{Edge- and Node-Disjoint Paths in P~Systems}
\def\authorrunning{M.J.~Dinneen, Y.-B.~Kim \& R.~Nicolescu}

\begin{document}
\maketitle


\begin{abstract}
In this paper, we continue our development of algorithms used for topological network discovery.  
We present native P~system versions of two fundamental problems in graph theory: 
finding the maximum number of edge- and node-disjoint paths between a source node and target node.  
We start from the standard depth-first-search maximum flow algorithms, 
but our approach is totally distributed, 
when initially no structural information is available
and each P~system cell has to even learn its immediate neighbors.
For the node-disjoint version, our P~system rules are designed to enforce node weight capacities (of one), 
in addition to edge capacities (of one), which are not readily available in the standard network flow algorithms.
\end{abstract}

Keywords: P~systems, P~modules, simple~P~modules, cell IDs, distributed algorithms, synchronous networks, breadth-first-search, depth-first-search, edge-disjoint paths, node-disjoint paths, network flow, network discovery, routing.


\section{Introduction}
\label{sec-introduction}

Inspired by the structure and interaction of living cells, P~systems provides a
distributed computational model, as introduced by G. P\u{a}un in
1998~\cite{Paun1998}.  The model was initially based on transition rules, but
was later expanded into a large family of related models, such as tissue and
neural P~systems (nP~systems) \cite{MartinVidePPR2003,Paun2002} and hyperdag
P~systems (hP~systems)~\cite{NDK-BWMC2009}.  Essentially, all versions of
P~systems have a structure consisting of cell-like membranes and a set of rules
that govern their evolution over time.  A large variety of rules have been used
to describe the operational behavior of P~systems, the main ones being:
multiset rewriting rules, communication rules and membrane handling rules.
Transition P~systems and nP~systems use multiset rewriting rules, P~systems
with symport/antiport operate by communicating immutable objects, P~systems
with active membranes combine all three type rules.  For a comprehensive
overview and more details, we refer the reader to \cite{Paun2002}.

Earlier in \cite{DKN-JLAP2010}, we have proposed an extensible framework 
called P~modules, to assist the programmability of P~systems.
P~modules enable the modular composition of complex P~systems
and also embrace the essential features of a variety of P~systems.
In this paper, we will use a restricted subset of this unifying model,
called simple~P~modules,
(subset equivalent to neural P~systems~\cite{MartinVidePPR2003}),
to develop algorithms for finding the maximum number of edge- and node-disjoint paths
between two cells in a fairly large class of P~systems, where 
duplex communication channels exist between neighboring cells.  
We assume that the digraph structure of the simple~P~module is completely unknown
(even the local neighboring cells must be discovered~\cite{NDK-WMC2009}) 
and we need to, via a distributed process, optimally create local routing tables 
between a given source and target cell.

There are endless natural applications that need to find alternative
routes between two points, from learning strategies to
neural or vascular remodeling after a stroke.
In this paper, we focus on a related but highly idealized goal, 
how to compute a maximum cardinality set of edge- and node-disjoint paths 
between two arbitrary nodes in a given digraph.

One obvious application related to networks is to find the best 
bandwidth utilization for routing of information between a 
source and target \cite{Robacker56}.  
For instance, streaming of applications for multi-core computations 
uses edge-disjoint paths routing for task decomposition and inter-task
communications \cite{SeoT-IPDPS2009}.
In fact, classical solutions are based on a network flow approach such as
given in \cite{FordF1956,EdmondsK1972}, 
or on Menger's Theorem, an old, but very useful, result, cited below.

\begin{theorem}[Menger~\cite{Menger1927}]
Let $D = (V,A)$ be a digraph and let $s, t \in V$. 
Then the maximum number of node-disjoint $s$--$t$ paths is equal 
to the minimum size of an $s$--$t$ disconnecting node set.
\end{theorem}

Another application is to find a maximum matching (or pairing) 
between two compatible sets such as the marriage arrangement problem or 
assigning workers to jobs.

Our third application (and a motivating problem for the authors) 
is the Byzantine Agreement problem \cite{DKN-JLAP2010,DKN-CMC2010},
in the case of non-complete graphs.
The standard solution (also based on Menger's Theorem) allows for $k$ faulty 
nodes (within a set of nodes of order at least $3k+1$) if and only if there
are at least $2k+1$ node-disjoint paths between each pair of nodes, 
to ensure that a distributed consensus can occur~\cite{Lynch1996}.


Briefly, the paper is organized as follows. In the next section, 
we give a formal definition of simple~P~modules, 
to give a unified platform for developing our P~systems algorithms.  
Next, in Section~\ref{sec-disjoint-paths} we summarize the standard network flow
approaches for finding edge- and node-disjoint paths in digraphs and 
we discuss optimizations and alternative strategies which are more appropriate for P~systems.
In Section~\ref{sec-structural-vs-search-digraph}, we discuss three possible
relations between the structural digraph underlying a simple~P~module
and the search digraph used for determining paths.
Section~\ref{sec-cell-neighborhoods} details breadth-first-search rules used 
to determine the local cell topologies,
i.e. all cell neighborhoods; this is a common preliminary phase for 
both the edge- and node-disjoint path implementations.
The next two sections detail depth-first-search rules
for the edge-disjoint case (in Section~\ref{sec-edge-disjoint-path-algorithm-rules})
and for the node-disjoint case (in Section~\ref{sec-node-disjoint-path-algorithm-rules}).
Finally, in Section~\ref{sec-conclusion}, 
we end with conclusions and some open problems.


\section{Preliminary}
\label{sec-preliminary}

We assume that the reader is familiar with the basic terminology 
and notations: functions, relations, graphs, edges, nodes (vertices), 
directed graphs, arcs, paths, directed acyclic graphs (dags), 
trees, alphabets, strings and multisets
\cite{NDK-IJCCC2010}.  We now introduce \emph{simple~P~modules},
as a unified model for representing several types of P~systems. 
Simple~P~modules are a simplified variety of the 
full P~modules, which omit the extensibility features
and use duplex communication channels only. 
With these restrictions, although their formal definitions are different,
simple~P~modules are essentially equivalent to neural P~systems~\cite{MartinVidePPR2003}. 
For the full definition of P~modules and
further details on recursive modular compositions,
the reader is referred to \cite{DKN-JLAP2010}.

\begin{definition}[simple~P~module]
\label{def:simple-P-module}
A \emph{simple~P~module} is a system $\Pi = (O, K, \delta)$%
, where:
\begin{enumerate}
  \item $O$ is a finite non-empty alphabet of \emph{objects};
  \item $K=\{\sigma_1,\sigma_2,\ldots,\sigma_n\}$ is a finite set of (internal) \emph{cells};
  \item $\delta$ is a binary relation on $K$,
      without reflexive or symmetric arcs,
	   which represents a set of parent-child structural arcs between existing cells, 
	   with \emph{duplex} communication capability.
\end{enumerate}

Each cell, $\sigma_i \in K$, 
has the initial form $\sigma_i = (Q_i, s_0, w_0, R_i)$ and 
general form $\sigma_i = (Q_i,s,w,R_i)$, where:  
\begin{itemize}
  \item $Q_i$ is a finite set of \emph{states};
  \item $s_0 \in Q_i$ is the \emph{initial state}; $s\in Q_i$ is the \emph{current
state};
  \item $w_0 \in {O}^*$ is the \emph{initial multiset} of objects;
   $w \in {O}^*$ is the \emph{current multiset} of objects;
  \item $R_i$ is a finite \emph{ordered} set of multiset rewriting \emph{rules} of the general form: 
        $s~x \rightarrow_{\alpha} s'~x'~(u)_{\beta_\gamma}$, 
        where $s,s'\in Q$, $x,x'\in {O}^*$, 
        $u \in {O}^*$, 
        $\alpha \in \{\modmin, \modmax\}$,
        $\beta \in \{\uparrow, \downarrow, \updownarrow\}$,
        $\gamma \in \{\modone, \modspread, \modrepl\} \cup K$.
        If $u=\lambda$, denoting the \emph{empty string} of objects, this rule can be 
abbreviated as $s~x \rightarrow_{\alpha} s'~x'$. The application of a rule takes two sub-steps, 
after which the cell's current state $s$ and multi-set of objects $x$ is replaced by $s'$ and $x'$,
respectively, while $u$ is a message which is sent as specified by the 
transfer operator $\beta_\gamma$.
\end{itemize}

The rules given by the ordered set(s) $R_i$ are applied in the \emph{weak priority} order~\cite{Paun2006}.
For a cell $\sigma_i = (Q_i, t, w, R_i)$, a rule $s~x \rightarrow_{\alpha}
s'~x'~(u)_{\beta_\gamma} \in R_i$
is \emph{applicable} if $t = s$ and $x \subseteq w$.
Additionally, if $s~x \rightarrow_{\alpha} s'~x'~(u)_{\beta_\gamma}$ is the first applicable rule,
then each subsequent applicable rule's target state 
(i.e. state indicated in the right-hand side) must be $s'$.
The semantics of the rules and the meaning of operators $\alpha$,
$\beta$, $\gamma$ are now described.
\end{definition}

For convenience, we will often identify a cell $\sigma_i$ with its index (or
\emph{cell ID}) $i$, when the context of the variable $i$ is clear.
We accept that cell IDs appear as objects or indices of complex objects.
Also, we accept \emph{custom cell ID} rules, which distinguish the cell ID of
the current cell from other cell IDs.
For example, the rule $0.1$ of Section~\ref{sec-cell-neighborhoods}, 
given as ``$s_0~ g_i \rightarrow_{\modmin} s_0$'' for cell $\sigma_i$,
appears as ``$s_0~ g_1 \rightarrow_{\modmin} s_0$'' in cell $\sigma_1$ 
and as ``$s_0~ g_2 \rightarrow_{\modmin} s_0$'' in cell $\sigma_2$.

The rewriting operator $\alpha = \modmax$
indicates that an applicable rewriting rule of $R_i$ is applied as many times as
possible, while the operator $\alpha = \modmin$ requires a rule of $R_i$ is 
applied only once.
The communication structure is based on the underlying digraph structure.  
In this paper, and we will only use the $\beta =\; \updownarrow$ 
and  $\gamma \in \{ \modrepl \} \cup K$ transfer operators.
With reference to cell $\sigma_i$, 
a rewriting rule using $(u)_{\updownarrow_{\modrepl}}$ indicates
that the multiset $u$ is replicated and sent to all neighboring cells
(parents and children), i.e. to all cells in $\delta(i) \cup \delta^{-1}(i)$. 
Assuming that cell $\sigma_j$ is a parent or a child of cell $\sigma_i$, 
i.e. $j \in \delta(i) \cup \delta^{-1}(i)$,
a rewriting rule using $(u)_{\updownarrow_j}$
indicates that the multiset $u$ is specifically sent cell $\sigma_j$.
Otherwise, if $j \notin \delta(i) \cup \delta^{-1}(i)$,
the rule is still applied, but the message $u$ is silently discarded.
The other non-deterministic transfer operators (e.g., $\modone$, $\modspread$,
$\uparrow$, $\downarrow$) are just mentioned here for completeness,
without details, and are not used in this paper.
For details, the interested reader is referred to \cite{DKN-JLAP2010}.

\begin{remark}
This definition of simple~P~module subsumes several earlier definitions of P~systems, hP~systems and nP~systems.
If $\delta$ is a \emph{tree}, then $\Pi$ is essentially a tree-based P~system 
(which can also be interpreted as a cell-like P~system).
If $\delta$ is a \emph{dag}, then $\Pi$ is essentially an hP~system.
If $\delta$ is a \emph{digraph}, then $\Pi$ is essentially an nP~system.
\end{remark}


\section{Disjoint paths in digraphs}
\label{sec-disjoint-paths}

We now briefly describe the basic edge- and node-disjoint paths algorithms, 
based on network flow, particularized for unweighted edges 
(i.e. all edge capacities are one), see Ford and Fulkerson~\cite{FordF1956}.
Our presentation will largely follow the standard approach, but also 
propose a couple of customizations and optimizations, specifically targeted for  
running on highly distributed and parallel computing models,
such as P~systems.

We are given a digraph $G=(V,E)$ and two nodes, 
a source node, $s \in V$, and a target node, $t \in V$.
We consider the following two optimization problems:
(1) find a maximum cardinality set of edge-disjoint paths from $s$ to $t$; and
(2) find a maximum cardinality set of node-disjoint paths from $s$ to $t$.
Obviously, any set of node-disjoint paths is also edge-disjoint,
but the converse is not true.
For example: 
\begin{itemize}
\item Figure~\ref{fig-node-edge-paths}~(a) 
shows a maximum cardinality set of edge-disjoint paths for a digraph $G$,
which is also a maximum cardinality set of node-disjoint paths,
\item Figure~\ref{fig-node-edge-paths}~(b) 
shows two maximum cardinality sets of edge-disjoint paths for the same digraph $G$,
which are not node-disjoint.
\item Figure~\ref{fig-different-edge-node} shows a digraph where 
the maximum number of edge-disjoint paths
is greater than the maximum number of node-disjoint paths.
\end{itemize}

\begin{figure}[h]
\centerline{\includegraphics[scale=1.0]{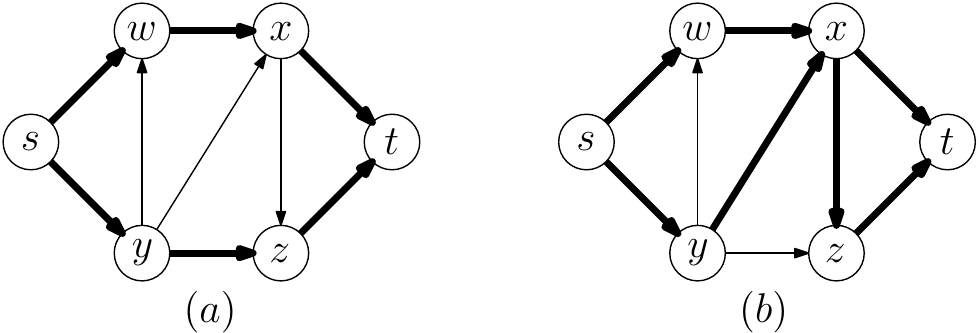}}
\caption{For this digraph, the maximum number of edge-disjoint paths from $s$ to $t$, 
which is 2, can be achieved in three ways: 
(a) paths set $\{ s.w.x.t$, $s.y.z.t\}$;
(b) either of the following two paths sets: 
$\{ s.w.x.t$, $s.y.x.z.t\}$, $\{ s.w.x.z.t$, $s.y.x.t\}$.
Paths shown in (a) are also node-disjoint, but paths shown in (b) are not.}
\label{fig-node-edge-paths}
\end{figure}

\begin{figure}[h]
\centerline{\includegraphics[scale=1.0]{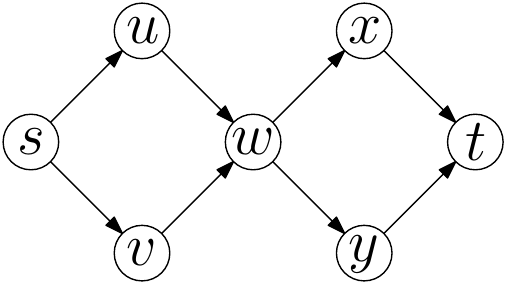}}
\caption{For this digraph, 
the maximum number of edge-disjoint paths from $s$ to $t$ (2) 
is greater than the maximum number of node-disjoint paths (1).}
\label{fig-different-edge-node}
\end{figure}


\subsection{Edge disjoint paths in digraphs}
\label{sec-edge-disjoint-paths}

In both edge- and node-disjoint cases, 
the basic algorithms work by repeatedly searching paths,
called augmenting paths, in an auxiliary structure, called residual network
or residual digraph.
We will first focus more on \emph{edge-disjoint paths}, 
because the \emph{node-disjoint paths}
can be considered as an edge-disjoint paths problem, with additional constraints. 

For the following ``network flow'' definition for digraphs with non-weighted
arcs, we say that an arc $(u,v)$ is in a set of paths $P$, denoted 
by the slightly abused notation $(u,v) \in P$, if there exists 
a path $\pi \in P$ that uses arc $(u,v)$.

\begin{definition}
Consider a digraph $G=(V,E)$, two nodes $s$ and $t$, $\{s, t\} \subseteq V$, 
and a set $P$ of edge-disjoint paths from $s$ to $t$.
Nodes in $P$ are called \emph{flow-nodes} and 
arcs in $P$ are called \emph{flow-arcs}.

Given path $\pi \in P$, each flow-arc $(u,v) \in \pi$ has a natural 
\emph{incoming} and \emph{outgoing} direction--the flow is from the source to the target;
with respect to $\pi$, $u$ is the \emph{flow-predecessor} of $v$ and 
$v$ is the \emph{flow-successor} of $u$.

The \emph{residual digraph} is the digraph $R=(V,E')$, 
where the arcs in $P$ are reversed, or, more formally,
$E'= (E \setminus \{(u,v) \mid (u,v) \in P\}) \cup \{(v,u) \mid (u,v) \in P\}$.
Any path from $s$ to $t$ in $R$ is called an \emph{augmenting path}.

Given augmenting path $\alpha$, each flow-arc $(u,v) \in \alpha$ has also a natural 
\emph{incoming} and \emph{outgoing} direction--the flow is from the source to the target;
with respect to $\alpha$, $u$ is the \emph{search-predecessor} of $v$ and 
$v$ is the \emph{search-successor} of $u$.
\end{definition}

\begin{fact}
\label{fact-augmenting}
Augmenting paths can be used to construct a larger set of edge-disjoint paths.
More precisely, consider a digraph $G$ and two nodes $s$ and $t$.
A set $P_k$ of $k$ edge-disjoint paths from $s$ to $t$
and an augmenting path $\alpha$ from $s$ to $t$
can be used together to construct a set $P_{k + 1}$ of $k + 1$ edge-disjoint paths.
First, paths in $\{ \alpha \} \cup P_k$ are fragmented, by 
removing ``conflicting'' arcs, i.e. arcs that appear in $Q \cup \tilde{Q}$,
where $Q = P \cap \tilde{\alpha}$ (where \~{} indicates arc reversal).
Then, new paths are created by concatenating resulting fragments.
For the formal definition of this construction,
we refer the reader to Ford and Fulkerson~\cite{FordF1956}.
Note that including a reversed arc in an augmenting path
is known as \emph{flow pushback operation}.
\end{fact}

This construction is illustrated in Figure~\ref{fig-residual-digraph}.
Figure~\ref{fig-residual-digraph}~(a) illustrates a digraph $G$ and 
a set $P_1$ of edge-disjoint paths from $s$ to $t$,
currently the singleton $\{ \pi_0 \}$, where $\pi_0 = s.y.x.t$.
Figure~\ref{fig-residual-digraph}~(b) shows its associated residual digraph $R$ 
(note the arcs reversal).
Figure~\ref{fig-residual-digraph}~(c) shows an augmenting path $\alpha$ in $R$, 
$\alpha = s.w.x.y.z.t$.
Figure~\ref{fig-residual-digraph}~(d) shows the extended set $P_2$ 
(after removing arcs $(x,y)$ and $(y,x)$),
consisting of two edge-disjoint paths from $s$ to $t$, 
$\pi_1 = s.w.x.t$, $\pi_2 = s.y.z.t$.
Figure~\ref{fig-residual-digraph1} shows a similar scenario, 
where another augmenting path is found.
Note that the two paths illustrated in Figure~\ref{fig-residual-digraph}~(d)
form both a maximum edge-disjoint set and a maximum node-disjoint set;
however, the two paths sets shown in Figure~\ref{fig-residual-digraph1}~(d)
form two other maximum edge-disjoint path sets, but none of them is node-disjoint.

\begin{figure}[h]
\centerline{\includegraphics[scale=0.85]{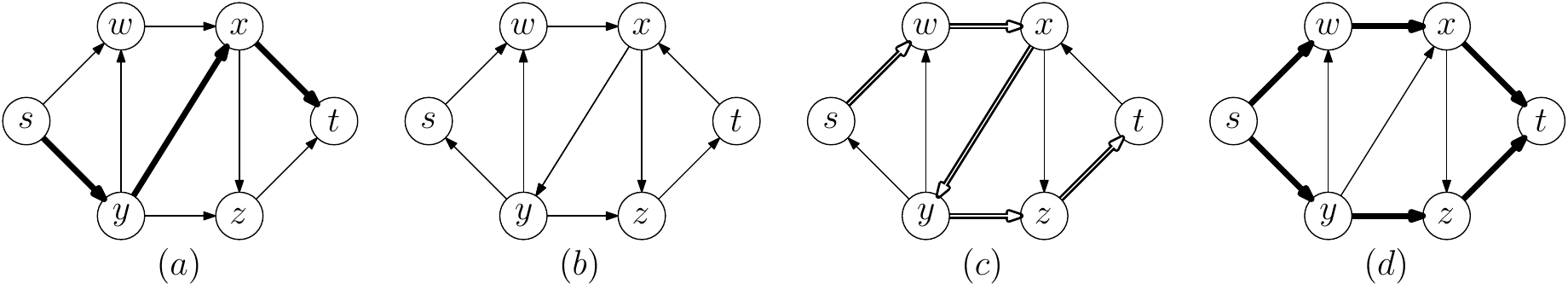}}
\caption{A residual digraph and an augmenting path: 
(a) a digraph $G$ and one (edge-disjoint) path $\pi_0$ from $s$ to $t$ (indicated by bold arrows).
(b) the residual digraph $R_0$ associated to digraph $G$ and path $\pi_0$.
(c) an augmenting path $\alpha$ in $R_0$  (indicated by hollow arrows).
(d) two new edge-disjoint paths $\pi_1$ and $\pi_2$,
reconstructed from $\pi_0$ and $\alpha$ (both indicated by bold arrows).}
\label{fig-residual-digraph}
\end{figure}

\begin{figure}[h]
\centerline{\includegraphics[scale=0.85]{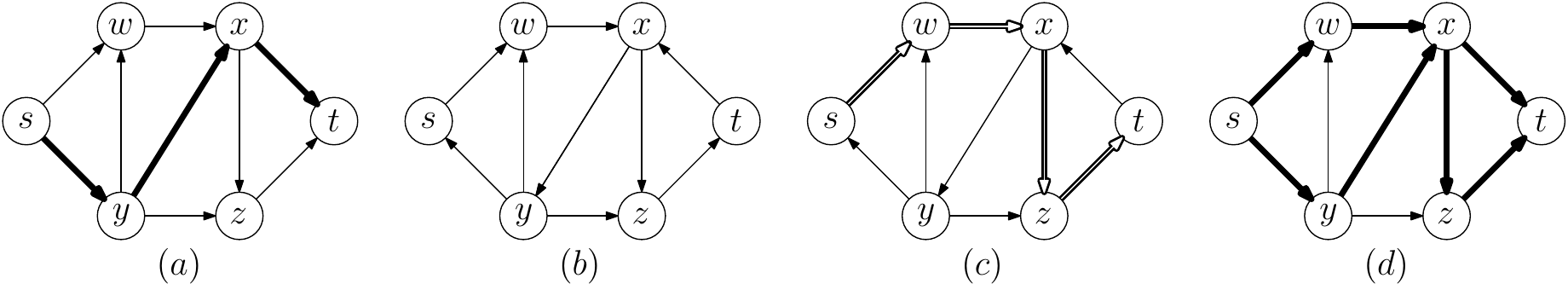}}
\caption{The residual digraph of Figure~\ref{fig-residual-digraph} 
with another augmenting path and two new paths sets, 
$\{s.w.x.t, s.y.x.z.t\}$, $\{s.w.x.z.t, s.y.x.t\}$,
which are edge-disjoint but not node-disjoint.}
\label{fig-residual-digraph1}
\end{figure}

\newpage

The pseudo-code of Algorithm~\ref{alg-basic-edge-disjoint-paths-algorithm} 
effectively finds the maximum number (and a representative set) 
of edge-disjoint paths from $s$ and $t$. 

\begin{algorithm}[\textbf{Basic edge-disjoint paths algorithm}]
\label{alg-basic-edge-disjoint-paths-algorithm}
~ 

\lstset{frame=none,xleftmargin={0.75 cm},numbers=left,numberstyle=\small}
\begin{lstlisting}
Input: #a digraph $G = (V,E)$ and two nodes $s \in V$, $t \in V$#
#$k = 0$ (the stage counter)#
#$P_0 = \emptyset$ (the current set of edge-disjoint paths)#
#$R_0 = G$ (the current residual digraph)#
loop
  #$\alpha = $ an augmenting path in $R_k$, from $s$ to $t$, if any (this is a search operation)#
  if $\alpha$ = null then break 
  #$k = k + 1$ (next stage)#
  #$P_{k} = $ the larger paths set constructed using $P_{k-1}$ and $\alpha$ (as indicated in Fact~\ref{fact-augmenting})#
  #$R_{k} = $ the residual digraph of $G$ and $P_{k}$#
end loop
Output: #$k$ and $P_k$, i.e. the maximum number of edge-disjoint paths and a representative set#
\end{lstlisting}

\end{algorithm}
\smallskip

Typically, the internal implementation of search at step~6 alternates between
a \emph{forward} mode in the residual digraph, 
which tries to extend a partial augmenting path,
and a backwards \emph{backtrack} mode in the residual digraph, 
which retreats from an unsuccessful attempt,
looking for other ways to move forward.
The internal implementation of step~9 (i.e. Fact~\ref{fact-augmenting}) walks backwards in the residual digraph,
as a \emph{consolidation} phase, which recombines the newly found augmenting path with the existing edge-disjoint paths.

This algorithm runs in $k+1$ stages, i.e. in up to $\mathtt{outdegree}(s) + 1$ stages,
if we count the number of times it looks for an augmenting paths,
and terminates when a new augmenting path is not found.
The actual procedure used (in step~6) to find the augmenting path separates two families of algorithms:
(1) algorithms from the Ford-Fulkerson family use a depth-first-search (DFS);
(2) algorithms from the Edmonds-Karp family use a breadth-first-search (BFS).
As usual, both DFS and BFS use ``bread crumb'' objects, as markers, to avoid cycles;
at the end of each stage, these markers are cleaned, to start again with a fresh context.
In this paper, we develop P~algorithms from the Ford-Fulkerson family, i.e. using DFS.


\subsection{Node disjoint paths in digraphs}
\label{sec-node-disjoint-paths}

The edge-disjoint version can be also used to find node-disjoint paths.
The textbook solution for the node-disjoint problem is usually 
achieved by a simple procedure which transforms the original digraph
in such a way that, on the transformed digraph,
the edge-disjoint problem is identical to 
the node-disjoint problem of the original digraph. 
Essentially, this procedure globally replaces every node $v$, 
other than $s$ and $t$, with two nodes,
an \emph{entry} node $v_{1}$ and an \emph{exit} node $v_{2}$,
connected by a single arc $(v_1,v_2)$. 
More formally, the new digraph $G' = (V', E')$ has
$V' = \{ s, t \} \cup \{ v_1, v_2 \mid v \in V \setminus \{s,t\}$, 
$E' = \{ (v_1, v_2) \mid v \in V \setminus \{s,t\} \} \cup \{ (u_2,v_1) \mid (u,v) \in E \}$,
where, for convenience, we assume that $s_1=s_2=s$ and $t_1=t_2=t$ are aliases.
This standard node-splitting technique is illustrated in Figure~\ref{fig-split-technique}.
It is straightforward to see that the newly introduced arcs
$(w_1,w_2)$, $(x_1,x_2)$, $(y_1,y_2)$ and $(z_1,z_2)$,
constrain any edge-disjoint solution to be also node-disjoint.

\begin{figure}[h]
\centerline{\includegraphics[scale=1.0]{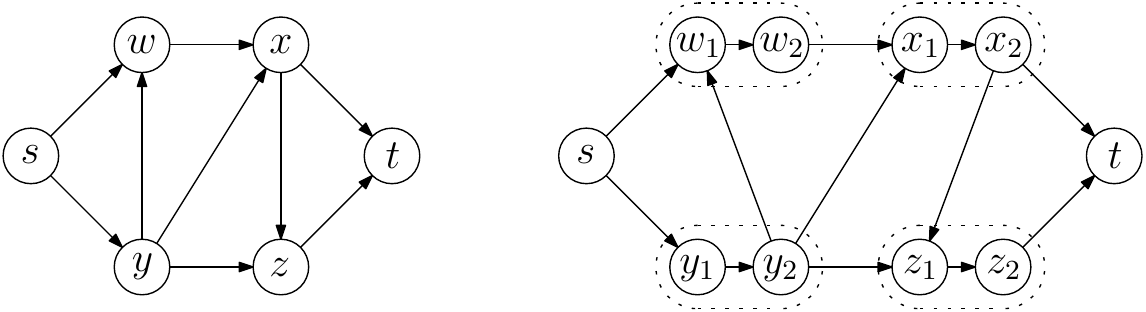}}
\caption{The node splitting technique.}
\label{fig-split-technique}
\end{figure}

However, in our case,
since each node is identified with a P~systems cell,
we cannot solve the node-disjoint paths problem
using the standard node-splitting technique.
We propose two non-standard search rules, 
which together limit the out-flow capacity of each $v \in V\setminus \{s,t\}$ to one,
by simulating the node-splitting technique, without actually splitting the nodes.
We believe that our rules can be used in other distributed network models where 
the standard node-splitting technique is not applicable.
These rules are illustrated by the scenario presented in Figure~\ref{fig-interesting-case},
where we assume that we have already determined a first flow-path, $s$.$x$.$y$.$z$.$t$,
and we are now trying to build a new augmenting path.

\begin{enumerate}
\item Consider the case when the augmenting path, consisting of $s$,
tries flow-node $y$ via the non-flow arc $(s,y)$.
We cannot continue with the existing non-flow arc $(y,t)$
(as the edge-disjoint version would do), 
because this will exceed node $y$'s capacity, which is one already.
Therefore, we continue the search 
with just the reversed flow-arc $(y,x)$.
Note that, in the underlying node-splitting scenario,
we are only visiting the entry node $y_1$, but not its exit pair $y_2$. 

\item Consider now the case when the augmenting path, extended now to $s$.$y$.$x$.$z$,
tries again the flow-node $y$, via the reversed flow-arc $(z,y)$.
It may appear that we are breaking the traditional search rules,
by re-visiting the already visited node $y$.
However, there is no infringement in the underlying node-splitting scenario,
where we are now trying the not-yet-visited exit node $y_2$
(to extend the underlying augmenting path $s$.$y_1$.$x_2$.$z_1$).
From $y$, we continue with any available non-flow arc, if any, 
otherwise, we backtrack. In our example, we continue with arc $(y,t)$.
We obtain a new augmenting ``path'', $s$.$y$.$x$.$z$.$y$.$t$
(corresponding to the underlying augmenting path $s$.$y_1$.$x_2$.$z_1$.$y_2$.$t$).
We further recombine it with the already existing flow-path $s$.$x$.$y$.$z$.$t$,
and we finally obtain the two possible node-disjoint paths,
$s$.$x$.$z$.$t$ and $s$.$y$.$t$.
\end{enumerate}

\begin{figure}[h]
\centerline{\includegraphics[scale=0.9]{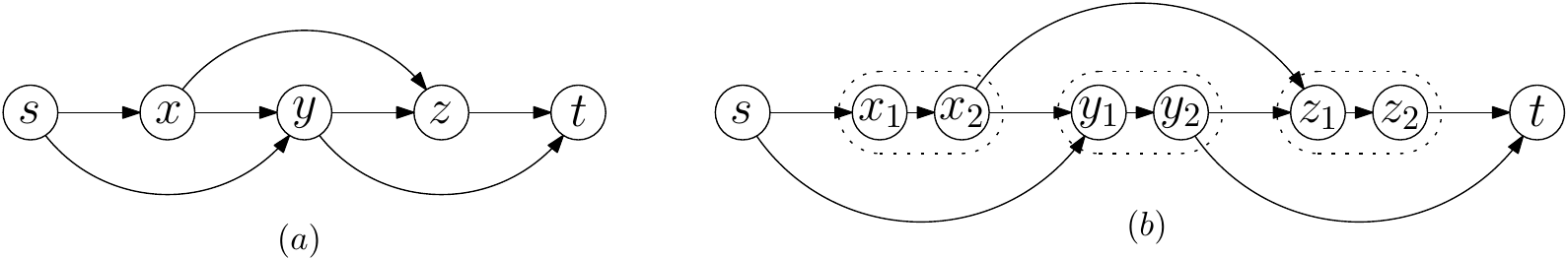}}
\caption{Node-disjoint paths.
(a) non-standard search: flow path $s$.$x$.$y$.$z$.$t$ and augmenting ``path'' $s$.$y$.$x$.$z$.$y$.$t$.
(b) node-splitting: flow path $s$.$x_1$.$x_2$.$y_1$.$y_2$.$z_1$.$z_2$.$t$ and 
augmenting path $s$.$y_1$.$x_2$.$z_1$.$y_2$.$t$.
}
\label{fig-interesting-case}
\end{figure}

The following theorem is now straightforward:

\begin{theorem}
If the augmented path search in step~6 of Algorithm~\ref{alg-basic-edge-disjoint-paths-algorithm} is modified as indicated above,
the algorithm will terminate with a restricted subset of edge-disjoint paths,
forming a maximum cardinal subset of node-disjoint paths.
\end{theorem}


\subsection{Pointer management}
\label{sec-pointer-management}

With respect to the implementation, the edge-disjoint version provides 
its own additional challenge, not present in the node-disjoint version. 
In the node-disjoint version,
each flow-node needs only one pointer to its flow-predecessor and another 
to its flow-successor. However, in the edge-disjoint version,
a flow-node can have $k$ flow-predecessors and $k$ flow-successors, with $k \geq 1$,
where each combination is possible, giving rise to $k!$ 
different edge-disjoint paths sets, each of size $k$, passing through this node.
A naive approach would require recording full details of all $k!$ possible 
size-$k$ paths sets,
or, at least, full details for one of them.

In our simplified approach, we do not keep full path details;
instead, a node needs only two size $k$ lists:
its flow-predecessors list and its flow-successors list.
Using this information, any of the actual $k!$ paths sets can be recreated on the fly,
by properly matching flow-predecessors with flow-successors.
As an example, consider node $x$ of Figure~\ref{fig-residual-digraph1}~(d),
which has two flow-predecessors, $w$ and $y$, and two flow-successors, $t$ and $z$;
thus $w$ is part of four distinct paths.
Node $w$ needs only two size-$k$ lists:
its flow-predecessors list, $\{ w, y \}$, and its flow-successors list $\{ z, t \}$.


\subsection{Possible optimization}
\label{sec-possible-optimization}

We propose a potential speed-up for Algorithm~\ref{alg-basic-edge-disjoint-paths-algorithm}.
We restrict this discussion to edge-disjoint paths and standard DFS;
however, the discussion can be generalized to more general flows and other search patterns.
Consider $V_s = E \cap (\{s\} \times V) = \{(s,v_1), (s,v_2), \dots (s,v_{d_s}) \}$, where $d_s = \mathtt{outdegree}(s)$. 
Step~6 systematically tries all arcs in $V_s$.
Without loss of generality, we assume that step~6 
always tries arcs in the order $(s,v_1), (s,v_2), \dots (s,v_{d_s})$,
stopping upon the first arc which is identified as starting an augmenting path $\alpha$,
say $(s,v_r)$, where $r \in [1,d_s]$.
For brevity, we will indicate this by saying that 
arc $(s,v_r)$ is the \emph{first} one that \emph{succeeds}
(and arcs $(s,v_1)$, $(s,v_2)$, \dots $(s,v_{r-1})$ \emph{fail}).

Consider a complete run of Algorithm~\ref{alg-basic-edge-disjoint-paths-algorithm}.
Assume that this algorithm finds $k$ augmenting paths and then stops. 
Assume that stage~$j \in [1,k]$, finds a new augmenting path $\alpha_j$,
which starts with arc $(s,v_{i_j})$, 
i.e. arc $(s,v_{i_j})$ is the first one that succeeds at stage~$j$.
A direct implementation of Algorithm~\ref{alg-basic-edge-disjoint-paths-algorithm} 
seems to require that step~6 starts a completely new search for each stage, 
restarting from $(s,v_1)$ and retrying arcs that have been previously tried 
(whether they failed or succeeded).

However, this is \emph{not} necessary.
Theorem~\ref{conjecture} indicates that stage~$j+1$ does not need to retry 
the nodes that have already been considered (whether they failed or succeeded).
Specifically, the indices indicating the successful arcs in $V_s$ are ordered by stage number,
$i_1 < i_2 < \dots < i_k$, and, at stage~$j \in [1,k]$, 
we can start the new search directly from arc $(s,v_{i_{j-1}+1})$
(where $i_0 = 0$). 

\begin{theorem}
\label{conjecture}
Step~6 of Algorithm~\ref{alg-basic-edge-disjoint-paths-algorithm} 
can start directly with the arc in $V_s$ 
which follows the previous stage's succeeding arc in $V_s$.
\end{theorem}

\begin{proof}
At stage $j \in [1,k]$, after finding augmenting path $\alpha_j$, 
the algorithm fragments $\alpha_i$ 
together with the previous set of edge-disjoint paths, $P_i$, 
deletes some arcs and reassembles a new and larger set of edge-disjoint-paths, $P_{i+1}$,
such that $| P_{i+1} | = | P_i | + 1$.

(1) Consider arc $(s,v_{i_j})$, the starting arc of path $\alpha_j$.
We first show that a successful arc, such as $(s,v_{i_j})$, 
need not be tried again by step~6, for two arguments:

(1.1) Arc $(s,v_{i_j})$ cannot start another augmenting path, 
because all following residual digraphs will only contain 
its reversal $(v_{i_j},s)$, never its direct form $(s,v_{i_j})$.

(1.2) Arc $(s,v_{i_j})$ cannot be revisited as part of another augmenting path.
No augmenting path contains arcs from $V \times \{s\}$, 
because the search (DFS or BFS) avoids already visited nodes (marked with ``pebbles''), 
and $s$ is always the starting point and thus the first node marked.
Therefore, once successful arc $(s,v_{i_j})$ is never deleted 
by ``flow pushback operations''.

From (1.1) and (1.2), we conclude that, once successful, an arc in $V_s$ 
will always be the starting arc of an edge-disjoint path
and need not be tried again by step~6.

\medskip

(2) We next show that, once failed, an arc in $V_s$ will always fail.
This part of the proof is by contradiction.
We select the first arc that succeeds after first failing and we exhibit a contradiction.
Consider that arc $(s,v_{i_g})$ is this arc, which succeeded at stage $g$,
but failed at least one earlier stage, and let $f$ be the earliest such stage, $f < g$.
It is straightforward to see that, in this case, 
$i_1 < i_2 < \dots < i_{f-1} < i_g < i_f < i_{f+1} < \dots < i_{g-1}$,
and arc $(s,v_{i_g})$ was tried and failed at all stages 
between $f$ (inclusive) and $g$.

As a thought experiment, let us stop the algorithm after step~$g$.
We have obtained $g$ augmenting paths, thus a set $P_g$, of $g$ edge-disjoint paths,
starting with arcs $(s,v_1)$, $(s,v_2)$, $\dots$, 
$(s,v_{f-1})$, $(s,v_f)$, $(s,v_{f+1})$, $\dots$, $(s,v_g)$.
Following the same though experiment, let us run the algorithm on 
digraph $G'$, obtained from $G$, by deleting all arcs in $V_s$ 
except arc $(s,v_{i_g})$ and those arcs that have been successful, 
before arc $(s,v_{i_g})$ was first tried and failed.
More formally, $G' = (V,E')$, where
$E' = E \setminus (V \setminus V'_s)$, 
$V'_s = \{(s,v_1)$, $(s,v_2)$, $\dots$, $(s,v_{f-1})$, $(s,v_{i_g})\}$.
Obviously, $| V'_s | = f $ and digraph $G'$ admits exactly $f$ edge-disjoint paths, 
because
(a) each of the remaining arcs in $V'_s$ can be the start 
of an edge-disjoint path in $P_g$ (which do not use any other arc in $V_s$), and  
(b) digraph $G'$ cannot admit more than $| V'_s | = f $ edge-disjoint paths.

It is straightforward to see that, Algorithm~\ref{alg-basic-edge-disjoint-paths-algorithm}, 
running on digraph $G'$, will follow exactly the same steps as running on digraph $G$, 
up to the point when it first fails on arc $(s,v_{i_g})$.
At this point, the run on digraph $G'$ stops, after finding $f$ augmenting paths
and constructing $f$ edge-disjoint paths.

Thus, the algorithm fails, because $f < g$, which contradicts its correctness.
Therefore, the algorithm will never succeed on an arc that has already failed and
never needs reconsidering again such arcs. 

This completes the second part of the proof. 
\end{proof}


\section{Structural and search digraphs in P~systems}
\label{sec-structural-vs-search-digraph}

In this section, we look at various way to reformulate 
the digraph edge- and node-disjoint path problems as a native P~system problem.
The P~system we consider is ``physically'' based on a digraph,
but this digraph is not necessarily the \emph{virtual} search digraph 
$G = (V, E)$, on which we intend to find edge- and node-disjoint paths.
Given a simple~P~system $\Pi = (O, K, \delta)$, where $\delta$ is its structural digraph,
we first identify cells as nodes of interest, $V \simeq K$. 
However, after that, we see three fundamentally distinct scenarios,
which differ in the way how the \emph{forward} and \emph{backward} modes 
(i.e. \emph{backtrack} and \emph{consolidation}) of Algorithm~\ref{alg-basic-edge-disjoint-paths-algorithm} map to the residual arcs and finally to the structural arcs. 

\begin{enumerate}
\item We set $E \simeq \delta$. In this case,
the forward mode follows the direction of parent-child arcs of $\delta$,
while the backward modes follow the reverse direction,
from child to parent.
\item We set $E \simeq \{ (v,u) \mid (u,v) \in \delta \}$. In this case,
the the backward modes follow the direction of parent-child arcs of $\delta$,
while the forward mode of the search follows the reverse direction,
from child to parent.
\item We set $E \simeq \{ (u,v), (v,u) \mid (u,v) \in \delta \}$. In this case,
the resulting search digraph is symmetric, 
and each of the arcs followed by the forward or backward modes of the search
can be either a parent-child arc in the original $\delta$ or its reverse.
\end{enumerate}

Cases (1) and (2) are simpler to develop. 
However, in this paper, we look for solutions in case (3), where 
all messages must be sent to all neighbors, parents and children together. 
Therefore, our rewriting rules use the $\beta = \updownarrow$ and
$\gamma \in \{ \modrepl \} \cup K$ transfer operators 
(also indicated in Section~\ref{sec-preliminary}).
Figure~\ref{fig-virtual-search-digraph} illustrates 
a simple~P~module and these three scenarios.

\begin{figure}[h]
\centerline{\includegraphics[scale=1.0]{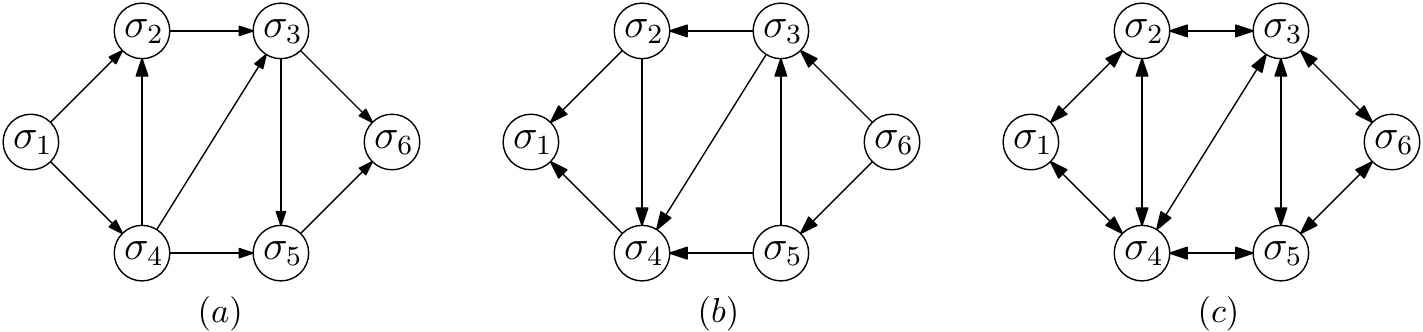}}
\caption{Three virtual search digraphs for the same simple~P~module.
(a) Same ``physical'' and search structure.
(b) The search structure reverses the ``physical'' structure.
(c) The search structure covers both the ``physical'' structure and its reverse.}
\label{fig-virtual-search-digraph}
\end{figure}

Note that, in any of the three cases, 
Algorithm~\ref{alg-basic-edge-disjoint-paths-algorithm} needs to be able
to follow both the parent-child and the child-parent directions of P~system structure.
Therefore, the structural arcs must support duplex communication channels.

After fixing the directions used by the virtual graph $G$, 
the next problem is to let the nodes identify their neighbors, 
i.e. discover the local network topology.


\section{Discovering cell neighbors}
\label{sec-cell-neighborhoods}

In this phase, cells discover their own neighbors.
Essentially, each cell sends its own ID to all its neighbors and
records the IDs sent from its neighbors.
This is a preliminary phase which is identical, 
for both edge-disjoint and node-disjoint versions.
Table~\ref{tab-neighbors} illustrates the immediate neighborhoods,
discovered at the end of this phase, for the P~system of
Figure~\ref{fig-virtual-search-digraph}~(a), with
the virtual search structure shown in 
Figure~\ref{fig-virtual-search-digraph}~(c). 

\setcounter{table}{\thefigure}
\begin{table}[ht] 
\caption{Neighbors table for the P~system of 
Figure~\ref{fig-virtual-search-digraph}~(a).
An object $n_j$ indicates that cell $\sigma_j$ is a neighbor of the current cell.}
\label{tab-neighbors}
\begin{center}
\begin{tabular}{ | c | c | c | } \hline
Cell & Neighbors & Objects \\ \hline 
$\sigma_1$ & $\{ \sigma_2, \sigma_4 \}$ & $\{ n_2, n_4 \}$ \\ \hline
$\sigma_2$ & $\{ \sigma_1, \sigma_3, \sigma_4 \}$ & $\{ n_1, n_3, n_4 \}$ \\ \hline
$\sigma_3$ & $\{ \sigma_2, \sigma_4, \sigma_5, \sigma_6 \}$ & $\{ n_2, n_4, n_5, n_6 \}$ \\ \hline
$\sigma_4$ & $\{ \sigma_1, \sigma_2, \sigma_3, \sigma_5 \}$ & $\{ n_1, n_2, n_3, n_5 \}$ \\ \hline
$\sigma_5$ & $\{ \sigma_3, \sigma_4, \sigma_6 \}$ & $\{ n_3, n_4, n_6 \}$ \\ \hline
$\sigma_6$ & $\{ \sigma_3, \sigma_5 \}$ & $\{ n_3, n_5 \}$ \\ \hline
\end{tabular}
\end{center}
\end{table} 
\addtocounter{figure}{1}

The set of objects used in this phase is 
$\{a, k, z\} \cup \bigcup_{1 \leq j \leq n} \{g_j, u_j, n_j\}$.
These objects have the following meanings:
$a$ indicates a cell reachable from $\sigma_s$;
$k$ is the marker of the source cell;
$z$ is the marker of the target cell;
$n_j$ indicates that $\sigma_j$ is a neighbor of the current cell;
$g_j$, $u_j$ indicate that $\sigma_j$ is the target cell;
$g_j$ only appears in the source cell, while $u_j$ does not have this restriction.

Initially, the source cell $\sigma_s$ has one copy of $g_j$, 
representing the ID of the target cell $\sigma_j$,
and the other cells are empty.
All cells start in state $s_0$.
Each reachable cell progresses through states $s_0, s_1, s_2, s_3, s_4$,
according to the rules given below.
In these generic rules (as elsewhere in this paper), we implicitly assume that
(1) subscript $i \in \{1,2,\dots n\}$ is customized for each cell to its cell ID;
and (2) subscript $j$ runs over all cell IDs ($j \in \{1,2,\dots n\}$),
effectively instantiating $n$ versions of each generic rule where it appears. 

\begin{tabular}[t]{ll}
  \begin{minipage}[t]{3.0in}
  \begin{enumerate}
  \setcounter{enumi}{-1}
  \item Rules for state $s_0$:
    \begin{enumerate}[1]
    \item $s_0~ g_i \rightarrow_{\modmin} s_0$
    \item $s_0~ g_j \rightarrow_{\modmin} s_1~ a k~ (u_j)_{\updownarrow_{\modrepl}}$
    \item $s_0~ u_i \rightarrow_{\modmin} s_1~ a z~ (u_i)_{\updownarrow_{\modrepl}}$
  
    \item $s_0~ u_i \rightarrow_{\modmax} s_1$
    \item $s_0~ u_j \rightarrow_{\modmin} s_1~ a~ (u_j)_{\updownarrow_{\modrepl}}$
    \end{enumerate}
  \end{enumerate}
  \end{minipage}

  \begin{minipage}[t]{3.0in}
  \begin{enumerate}
  \setcounter{enumi}{0}
  \item Rules for state $s_1$:
    \begin{enumerate}[1]
    \item $s_1~ a \rightarrow_{\modmin} s_2~ a~ (n_i)_{\updownarrow_{\modrepl}}$
    \end{enumerate}
  \item Rules for state $s_2$:
    \begin{enumerate}[1]
    \item $s_2~ a \rightarrow_{\modmin} s_3~ a$
    \end{enumerate}
  \item Rules for state $s_3$:
    \begin{enumerate}[1]
    \item $s_3~ a \rightarrow_{\modmin} s_4~ a$
    \item $s_3~ u_j \rightarrow_{\modmax} s_4$
    \end{enumerate}
  \end{enumerate}
  \end{minipage}
\end{tabular}

~\

The following example indicates how our generic rules are instantiated to take account  the cell IDs, more specifically, how rules $0.1$ and $0.2$ are instantiated in cell $\sigma_1$:
\begin{itemize}
\item $s_0~ g_1 \rightarrow_{\modmin} s_0$ 
\item $s_0~ g_1 \rightarrow_{\modmin} s_1~ a k~ (u_1)_{\updownarrow_{\modrepl}}$
\item $s_0~ g_2 \rightarrow_{\modmin} s_1~ a k~ (u_2)_{\updownarrow_{\modrepl}}$
\item $\dots$, 
\item $s_0~ g_n \rightarrow_{\modmin} s_1~ a k~ (u_n)_{\updownarrow_{\modrepl}}$
\end{itemize}

The state transitions performed by cell $\sigma_i$, $i \in \{1, 2, \ldots, n\}$, 
are briefly discussed below.

\begin{itemize}
\item $s_0 \rightarrow s_1$:
      If $\sigma_i$ contains $g_j$, $\sigma_i$ becomes the \emph{source cell}.
      The source cell broadcasts object $u_j$ to all its neighbors.
      After receiving an object $u_j$, 
      cell $\sigma_i$ becomes either the \emph{target cell}, if $i=j$;
      or an \emph{intermediate cell}, otherwise.
      Further, each cell relays one of received $u_j$ objects to all its neighbors.

\item $s_1 \rightarrow s_2$: 
      Cell $\sigma_i$ broadcasts $n_i$ to all its neighbors.
      Additionally, $\sigma_i$ accumulates $n_j$ objects from neighbors.

\item $s_2 \rightarrow s_3$: 
      Cell $\sigma_i$ accumulates further $n_j$ objects from neighbors.

\item $s_3 \rightarrow s_4$: 
      Cell $\sigma_i$ accumulates further $n_j$ objects from neighbors.
      Additionally, $\sigma_i$ removes superfluous $e_j$ objects.
\end{itemize}


\section{Simple~P~module rules for edge-disjoint paths algorithm}
\label{sec-edge-disjoint-path-algorithm-rules}

First, we give a simple~P~module specification of the 
edge-disjoint paths algorithm presented in Section~\ref{sec-disjoint-paths}.  
We explicitly state our problem in terms of expected input and output.
We need to compute a set of edge-disjoint paths of maximum cardinality
between given source and target cells. 

\smallskip

\noindent \textbf{Edge-Disjoint Paths Problem}\\
\noindent \textbf{Input:} A simple~P~module $\Pi= (O, K, E, \delta)$, 
where the source cell $\sigma_s \in K$ contains a token $t_t$ identifying
the ID of the target cell $\sigma_t \in K$.\\
\noindent \textbf{Output:}  If $s \neq t$, each cell $\sigma_i \in K$
contains a set of objects $P_i = \{ p_j \mid (j,i) \mbox{ is a flow-arc}\}$ and  
a set of objects $C_i = \{ c_j \mid (i,j) \mbox{ is a flow-arc}\}$ that 
represent a maximum set of edge-disjoint paths from $\sigma_s$ to $\sigma_t$, 
where the following constraints hold:\\

\begin{minipage}{5.5in}
\begin{description}
\item[flow-arcs:] 
     $c_i \not\in C_i$, $p_i \not\in P_i$, 
     $c_j \in C_i \Leftrightarrow p_i \in P_j$ and 
     $c_j \in C_i \Rightarrow j \in \delta(i) \cup \delta^{-1}(i)$.
\item[source and target:] 
     $P_s = \emptyset$ and $C_t = \emptyset$.
\item[in flow = out flow:] If $i\not\in \{s,t\}$ then $|C_i|=|P_i|$. 
\item[only paths:] 
With $S(I)= \bigcup_{i \in I} \{j \mid  j \in C_i\}$,
$S^{n-1}(I)=S(S(\cdots S(I)\cdots))=\emptyset$.

\end{description}
\end{minipage}

\smallskip

\noindent Because of the network flow properties, we must also have $|C_s|=|P_t|$,
which also represents the maximum number of edge-disjoint paths.

This implementation has two phases: Phase~I, which is the discovery phase
described in Section~\ref{sec-cell-neighborhoods} (using states $s_0$ to $s_4$), 
and Phase~II, described below (which starts in state $s_4$ and ends in state $s_{13}$).
Table~\ref{tab-edge-output-objects} illustrates
the expected algorithm output, for a simple~P~module with the cell structure
corresponding to Figure~\ref{fig-virtual-search-digraph}~(a).
For convenience, although these are deleted near the algorithm's end,
we also list all the local neighborhood objects 
$N_i = \{n_j \mid j\in \delta(i) \cup \delta^{-1}(i)\}$, for $i \in \{1, 2, \ldots, n\}$, 
which are determined in Phase~I.

\setcounter{table}{\thefigure}
\begin{table}[ht]
\caption{A representation of maximum edge-disjoint paths for simple~P~module of
Figure~\ref{fig-virtual-search-digraph}~(a).}
\label{tab-edge-output-objects}
\begin{center}
\begin{tabular}{ | c | c | c | c | } \hline
Cell $\backslash$ Objects & $N_i$ & $P_i$ & $C_i$ \\ \hline 
$\sigma_1$ & $\{n_2, n_3\}$ & $\emptyset$ & $\{c_2, c_4\}$ \\ \hline
$\sigma_2$ & $\{n_1, n_3, n_4\}$ & $\{p_1\}$ & $\{c_3\}$ \\ \hline
$\sigma_3$ & $\{n_2, n_4, n_5, n_6\}$ & $\{p_2, p_4\}$ & $\{c_5, c_6\}$ \\ \hline
$\sigma_4$ & $\{n_1, n_2, n_3, n_5\}$ & $\{p_1\}$ & $\{c_3\}$ \\ \hline
$\sigma_5$ & $\{n_3, n_4, n_6\}$ & $\{p_3\}$ & $\{c_6\}$ \\ \hline
$\sigma_6$ & $\{n_3, n_5\}$ & $\{p_3, p_5\}$ & $\emptyset$ \\ \hline
\end{tabular}
\end{center}
\end{table}
\addtocounter{figure}{1}

In addition to the set of objects used in Phase~I,
this phase uses the following set of objects:
$\{ b_j$, $c_j$, $d_j$, $e_j$, $f_j$, $h_j$, $m_j$, 
$p_j$, $q_j$, $r_j$, $t_j$, $x_j$, $y_j \} \cup \{ v, w \}$.
In cell $\sigma_i$, these objects have the following meanings:
\begin{itemize}
\item $b_j$ indicates a pushback, received from $\sigma_i$'s flow-successor $\sigma_j$;

\item $e_j$ records the sender of a pushback, if $\sigma_i$ is not-yet-visited;

\item $h_j$ records the sender of a pushback, if $\sigma_i$ has already been visited;

\item $c_j$ indicates that $\sigma_j$ is $\sigma_i$'s flow-successor;

\item $p_j$ indicates that $\sigma_j$ is $\sigma_i$'s flow-predecessor;

\item $r_j$ records a pushback sent by $\sigma_i$ to its flow-predecessor $\sigma_j$;

\item $t_j$ records a backtrack request, after a failed pushback to $\sigma_j$;

\item $d_j$ indicates that $\sigma_j$ is $\sigma_i$'s search-successor;

\item $q_j$ indicates that $\sigma_j$ is $\sigma_i$'s search-predecessor;

\item $f_j$ indicates an attempted search extension received from $\sigma_j$;

\item $m_j$ records that $\sigma_j$ was unsuccessfully tried;

\item $x_j$ indicates that $\sigma_j$ rejected a flow-extension or flow-pushback attempt;

\item $x_j$ indicates that $\sigma_j$ accepted a flow-extension or flow-pushback attempt;

\item $v$ requests that $\sigma_i$ resets the record of all tried and visited neighbors;

\item $w$ requests that $\sigma_i$ remains idle for one step.
\end{itemize}

Initially, the source cell $\sigma_s$ has one copy of $g_j$, 
representing the ID of the target cell $\sigma_j$,
and the other cells are empty.
All cells start in state $s_0$.
According to the rules of Phase~I, 
each reachable cell progresses to state $s_4$,
which is the start of Phase~II, whose generic rules are given below.
In these rules (as in Phase~I), we implicitly assume that
(1) subscript $i \in \{1,2,\dots n\}$ is customized for each cell to its cell ID;
and (2) subscripts $j, k$ run over all cell IDs, $j, k \in \{1,2,\dots n\}$,
and $j \neq k$.
To apply the \emph{optimization} proposed in Section~\ref{sec-possible-optimization}, 
replace rule 5.3 ``$s_5~ d_j x_j \rightarrow_{\modmin} s_5~ a m_j$'' by 
``$s_5~ d_j x_j \rightarrow_{\modmin} s_5~ a$''.

\begin{tabular}[t]{ll}
  \begin{minipage}[t]{3.0in}
  \begin{enumerate}

  \setcounter{enumi}{3}
  \item Rules for a cell $\sigma_i$ in state $s_4$:
    \begin{enumerate}[1]
    \item $s_4~ k \rightarrow_{\modmin} s_5~ k$

    \item $s_4~ z \rightarrow_{\modmin} s_6~ z$

    \item $s_4~ a \rightarrow_{\modmin} s_7~ a$
    \end{enumerate}

  \item Rules for a cell $\sigma_i$ in state $s_5$:
    \begin{enumerate}[1]
    \item $s_5~ a n_j \rightarrow_{\modmin} s_5~ d_j~ (f_i)_{\updownarrow_{j}}$

    \item $s_5~ d_j y_j \rightarrow_{\modmin} s_{12}~ a c_j w w~ (v)_{\updownarrow_\modrepl}$ 
    \item $s_5~ d_j x_j \rightarrow_{\modmin} s_5~ a m_j$ 

    \item $s_5~ b_j \rightarrow_{\modmin} s_5~ (x_i)_{\updownarrow_{j}}$
    \item $s_5~ f_j \rightarrow_{\modmin} s_5~ (x_i)_{\updownarrow_{j}}$ 

    \item $s_5~ a k \rightarrow_{\modmin} s_{13}~ a a w w~ (a)_{\updownarrow_{\modrepl}}$
    \end{enumerate}

  \item Rules for a cell $\sigma_i$ in state $s_6$:
    \begin{enumerate}[1]
    \item $s_6~ n_j f_j \rightarrow_{\modmin} s_6~ p_j~ (y_i)_{\updownarrow_{j}}$ 

    \item $s_6~ v \rightarrow_{\modmin} s_{12}~ w w~ (v)_{\updownarrow_{\modrepl}}$

    \item $s_6~ a a z \rightarrow_{\modmin} s_{13}~ a a w w~ (a)_{\updownarrow_{\modrepl}}$
    \end{enumerate}

  \item Rules for a cell $\sigma_i$ in state $s_7$:
    \begin{enumerate}[1]

    \item $s_7~ v \rightarrow_{\modmin} s_{12}~ w w~ (v)_{\updownarrow_\modrepl}$
    \item $s_7~ a a \rightarrow_{\modmin} s_{13}~ a a w w~ (a)_{\updownarrow_{\modrepl}}$

    \item $s_7~ n_j f_j \rightarrow_{\modmin} s_8~ q_j$
    \item $s_7~ c_j b_j \rightarrow_{\modmin} s_8~ e_j$

    \item $s_7~ h_j \rightarrow_{\modmin} s_{11}~ c_j~ (x_i)_{\updownarrow_j}$
    \item $s_7~ p_j q_k \rightarrow_{\modmin} s_{10}~ p_j q_k$

    \item $s_7~ q_j \rightarrow_{\modmin} s_7~ m_j~ (x_i)_{\updownarrow_j}$
    \item $s_7~ f_j \rightarrow_{\modmin} s_7~ (x_i)_{\updownarrow_j}$
    \end{enumerate}

  \item Rules for a cell $\sigma_i$ in state $s_8$:
    \begin{enumerate}[1]
    \item $s_8~ a n_j \rightarrow_{\modmin} s_9~ a d_j~ (f_i)_{\updownarrow_j}$

    \item $s_8~ a \rightarrow_{\modmin} s_{10}~ a$ 
    \end{enumerate}

  \end{enumerate}
  \end{minipage}
&
  \begin{minipage}[t]{3.0in}
  \begin{enumerate}

  \setcounter{enumi}{8}
  \item Rules for a cell $\sigma_i$ in state $s_9$:
    \begin{enumerate}[1]
    
    \item $s_9~ d_j y_j e_k \rightarrow_{\modmin} s_7~ c_j m_k~ (y_i)_{\updownarrow_k}$
    \item $s_9~ d_j y_j q_k \rightarrow_{\modmin} s_7~ c_j p_k~ (y_i)_{\updownarrow_k}$

    \item $s_9~ d_j x_j \rightarrow_{\modmin} s_8~ m_j$

    \item $s_9~ c_j b_j \rightarrow_{\modmin} s_9~ c_j~ (x_i)_{\updownarrow_j}$
    \item $s_9~ n_j f_j \rightarrow_{\modmin} s_9~ m_j~ (x_i)_{\updownarrow_j}$ 
    \end{enumerate} 
  
  \item Rules for a cell $\sigma_i$ in state $s_{10}$:
    \begin{enumerate}[1]
    \item $s_{10}~ a p_j \rightarrow_{\modmin} s_{11}~ a r_j~ (b_i)_{\updownarrow_{j}}$ 

    \item $s_{10}~ a e_j \rightarrow_{\modmin} s_7~ a c_j~ (x_i)_{\updownarrow_{j}}$ 
    \item $s_{10}~ a q_j \rightarrow_{\modmin} s_7~ a m_j~ (x_i)_{\updownarrow_{j}}$ 
    \end{enumerate}

  \item Rules for a cell $\sigma_i$ in state $s_{11}$:
    \begin{enumerate}[1]
    \item $s_{11}~ r_j y_j e_k \rightarrow_{\modmin} s_7~ m_j m_k~ (y_i)_{\updownarrow_k}$
    \item $s_{11}~ r_j y_j q_k \rightarrow_{\modmin} s_7~ m_j p_k~ (y_i)_{\updownarrow_k}$
    \item $s_{11}~ r_j x_j \rightarrow_{\modmin} s_{10}~ t_j$

    \item $s_{11}~ c_j b_j \rightarrow_{\modmin} s_7~ h_j$
    \item $s_{11}~ n_j f_j \rightarrow_{\modmin} s_{11}~ m_j~ (x_i)_{\updownarrow_j}$ 
    \end{enumerate}

  \item Rules for a cell $\sigma_i$ in state $s_{12}$:
    \begin{enumerate}[1]
    \item $s_{12}~ w \rightarrow_{\modmin} s_{12}$
    \item $s_{12}~ v \rightarrow_{\modmax} s_{12}$

    \item $s_{12}~ m_j \rightarrow_{\modmin} s_{12}~ n_j$
    \item $s_{12}~ t_j \rightarrow_{\modmin} s_{12}~ p_j$

    \item $s_{12}~ k \rightarrow_{\modmin} s_5~ k$
    \item $s_{12}~ z \rightarrow_{\modmin} s_6~ z$
    \item $s_{12}~ a \rightarrow_{\modmin} s_7~ a$
    \end{enumerate}

  \item Rules for a cell $\sigma_i$ in state $s_{13}$:
    \begin{enumerate}[1]
    \item $s_{13}~ w \rightarrow_{\modmin} s_{13}$

    \item $s_{13}~ a \rightarrow_{\modmax} s_0$

    \item $s_{13}~ t_j \rightarrow_{\modmin} s_0~ p_j$
    \item $s_{13}~ n_j \rightarrow_{\modmin} s_0$
    \item $s_{13}~ m_j \rightarrow_{\modmin} s_0$

    \end{enumerate}

  \end{enumerate}
  \end{minipage}
\end{tabular}

The following paragraphs describe details of several critical steps of our edge-disjoint algorithm, such as forward and consolidation modes of intermediate cells.

\begin{itemize}

\item The rules in state $s_8$ cover the 
      forward mode attempt of an intermediate cell $\sigma_i$,
      via a non-flow arc.

      If $\sigma_i$ has a not-yet-tried neighbor $\sigma_h \notin P_i \cup C_i$, 
      the search continues with $\sigma_h$;
      otherwise the search backtracks.
      

\item The rules in state $s_9$ cover the 
      consolidation mode for an intermediate cell $\sigma_i$,
      who has succeeded a forward extension on a non-flow arc.

      During the consolidation process, the behavior of $\sigma_i$ depends on 
      whether $\sigma_i$ was reached by a flow arc or non-flow arc:
      \begin{itemize}
      \item If $\sigma_i$ was reached by flow arc $(i,j)$, in a reverse direction,
            $\sigma_i$ replaces its flow-successor $\sigma_j$ with $\sigma_h$,
            where $\sigma_h$ is its search-predecessor.
      \item If $\sigma_i$ was reached by a non-flow arc $(k,i)$,
            $\sigma_i$ sets $\sigma_k$ as a flow-predecessor and 
            $\sigma_h$ as a flow-successor,
            where $\sigma_h$ and $\sigma_j$ are $\sigma_i$'s 
            search-predecessor and search-successor, respectively.
      \end{itemize}
      
\item The rules in state $s_{10}$ cover the 
      forward mode attempt of an intermediate cell $\sigma_i$,
      via a pushback.

      If $\sigma_i$ has a not-yet-tried flow-predecessor $\sigma_h$,
      the search continues with $\sigma_h$
      (i.e. use flow arc $(h,i)$, in a reverse direction);
      otherwise the search backtracks.

\item The rules in state $s_{11}$ cover the 
      consolidation mode for an intermediate cell $\sigma_i$,
      who has succeeded a forward extension via a pushback to $\sigma_h$
      (i.e. on a flow-arc $(h,i)$, in a reverse direction).

      During the consolidation process, the behavior of $\sigma_i$ depends on 
      whether $\sigma_i$ was reached by a flow arc or non-flow arc:
      \begin{itemize}
      \item If $\sigma_i$ was reached by flow arc $(i,j)$, in a reverse direction,
            $\sigma_i$ removes its flow-predecessor $\sigma_h$ and 
            flow-successor $\sigma_j$.
      \item If $\sigma_i$ was reached by a non-flow arc $(k,i)$,
            $\sigma_i$ replaces its flow-predecessor $\sigma_h$ with $\sigma_k$,
            where $\sigma_k$ is $\sigma_i$'s search-predecessor.
      \end{itemize}
\end{itemize}


\begin{theorem}
For a simple~P~module with $n$ cells and $m = |\delta|$ edges,
the algorithm in this section runs in $O(mn)$ steps.
\end{theorem}


\section{Simple~P~module rules for node-disjoint paths algorithm}
\label{sec-node-disjoint-path-algorithm-rules}

First, we give a simple~P~module specification of the 
node-disjoint paths algorithm presented in Section~\ref{sec-disjoint-paths}.  
We explicitly state our problem in terms of expected input and output.
We need to compute a set of node-disjoint paths of maximum cardinality
between given source and target cells. 

\medskip

\noindent \textbf{Node-Disjoint Paths Problem}\\
\noindent \textbf{Input:} A simple~P~module $\Pi= (O, K, E, \delta)$, 
where the source cell $\sigma_s \in K$ contains a token $t_t$ identifying
the ID of the target cell $\sigma_t \in K$.\\
\noindent \textbf{Output:}  If $s\not=t$, each cell $\sigma_i \in K$
contains a set of objects $P_i = \{ p_j \mid (j,i) \mbox{ is a flow-arc}\}$ and  
a set of objects $C_i = \{ c_j \mid (i,j) \mbox{ is a flow-arc}\}$ that 
represent a maximum set of
node-disjoint paths from $\sigma_s$ to $\sigma_t$ where the following constraints hold:\\

\begin{minipage}{5.5in}
\begin{description}
\item[flow-arcs:] 
     $c_i \not\in C_i$, $p_i \not\in P_i$, 
     $c_j \in C_i$ $\Leftrightarrow$ $p_i \in P_j$ and 
     $c_j \in C_i \Rightarrow j \in \delta(i) \cup \delta^{-1}(i)$.
\item[source and target:] 
     $P_s=\emptyset$ and $C_t=\emptyset$.
\item[node-disjoint:] 
     If $i\not\in \{s,t\}$ then $|C_i|=|P_i| \leq 1$. 
\item[only paths:] 
     With $S(i)=\left\{ 
     \begin{array}{cl}
     t & \mbox{if } i\in \{s,t\} \mbox{ or } |C_i|=0\\
     j & \mbox{when } C_i=\{c_j\}
     \end{array}
     \right\}$, $S^{n-1}(i)=S(S(\cdots S(i)\cdots))=t$.
\end{description}
\end{minipage}

\medskip

\noindent Because of the network flow properties, we must also have $|C_s|=|P_t|$,
which also represents the maximum number of node-disjoint paths.
Notice the constraints to require only paths has been simplified in that the
successor $S(i)$ of non-source cell $\sigma_i$ is a single cell instead of a set of 
cells that was needed for the general edge-disjoint problem.

\bigskip

Table~\ref{tab-node-output-objects} illustrates
the expected algorithm output, for a simple~P~module with the cell structure
corresponding to Figure~\ref{fig-virtual-search-digraph}~(a).
For convenience, although these are deleted near the algorithm's end,
we also list all the local neighborhood objects 
$N_i = \{n_j \mid j\in \delta(i) \cup \delta^{-1}(i)\}$, for $i \in \{1, 2, \ldots, n\}$, 
which are determined in Phase~I.

\setcounter{table}{\thefigure}
\begin{table}[ht]
\caption{A representation of maximum node-disjoint paths for simple~P~module of
Figure~\ref{fig-virtual-search-digraph}~(a).}
\label{tab-node-output-objects}
\begin{center}
\begin{tabular}{ | c | c | c | c | } \hline
Cell $\backslash$ Objects & $N_i$ & $P_i$ & $C_i$ \\ \hline 
$\sigma_1$ & $\{n_2, n_3\}$ & $\emptyset$ & $\{c_2, c_4\}$ \\ \hline
$\sigma_2$ & $\{n_1, n_3, n_4\}$ & $\{p_1\}$ & $\{c_3\}$ \\ \hline
$\sigma_3$ & $\{n_2, n_4, n_5, n_6\}$ & $\{p_2\}$ & $\{c_6\}$ \\ \hline
$\sigma_4$ & $\{n_1, n_2, n_3, n_5\}$ & $\{p_1\}$ & $\{c_5\}$ \\ \hline
$\sigma_5$ & $\{n_3, n_4, n_6\}$ & $\{p_4\}$ & $\{c_6\}$ \\ \hline
$\sigma_6$ & $\{n_3, n_5\}$ & $\{p_3, p_5\}$ & $\emptyset$ \\ \hline
\end{tabular}
\end{center}
\end{table}
\addtocounter{figure}{1}

The rules of this node-disjoint path algorithm are exactly 
the rules of the edge-disjoint path algorithm described in 
Section~\ref{sec-edge-disjoint-path-algorithm-rules},
where the rules for state $s_7$ are replaced by the following group of rules.

  \begin{enumerate}
  \setcounter{enumi}{6}
  \item Rules for a cell $\sigma_i$ in state $s_7$:
    \begin{enumerate}[1]

    \item $s_7~ v \rightarrow_{\modmin} s_{12}~ w w~ (v)_{\updownarrow_\modrepl}$
    \item $s_7~ a a \rightarrow_{\modmin} s_{13}~ a a w w~ (a)_{\updownarrow_{\modrepl}}$

    \item $s_7~ n_j f_j p_k \rightarrow_{\modmin} s_{10}~ q_j p_k$
    \item $s_7~ n_j f_j \rightarrow_{\modmin} s_8~ q_j$
    \item $s_7~ c_j b_j \rightarrow_{\modmin} s_8~ e_j$

    \item $s_7~ h_j \rightarrow_{\modmin} s_8~ e_j$
    \item $s_7~ r_j \rightarrow_{\modmin} s_{11}~ r_j$

    \item $s_7~ q_j \rightarrow_{\modmin} s_7~ m_j~ (x_i)_{\updownarrow_j}$
    \item $s_7~ f_j \rightarrow_{\modmin} s_7~ (x_i)_{\updownarrow_j}$
    \end{enumerate}

  \end{enumerate}

The new state $s_7$ rules implement our proposed non-standard technique 
described in Section~\ref{sec-node-disjoint-paths}, 
for enforcing node capacities to one, 
without node-splitting.


\medskip

The running time of our node-disjoint paths algorithm runs in
polynomial number of steps, since the algorithm is 
direct implementation of Ford-Fulkerson's network flow algorithm.

\begin{theorem}
For a simple~P~module with $n$ cells and $m = |\delta|$ edges,
the algorithm in this section runs in $O(mn)$ steps.
\end{theorem}


\section{Conclusion and Open Problems}
\label{sec-conclusion}

Using the newly introduced simple~P~modules framework,
we have presented native P~system versions 
of the edge- and node-disjoint paths problems.
We have started from standard network flow ideas,
with additional constraints required by our model,
e.g., cells that start without any knowledge about the local and global structure. 
Our P~algorithms use a depth-first search technique
and iteratively build routing tables, 
until they find the maximum number of disjoint paths;
Our P~algorithms run in polynomial time, 
comparable to the standard versions of the Ford-Fulkerson algorithms.   
We have proved and used a speedup optimization,
which probably was not previously known.
For node-disjoint paths, we proposed an alternate set of search rules, 
which can be used for other synchronous network models, where, as in P~systems, 
the standard node-splitting technique is not applicable.

All of our previous P~algorithms assumed that the structural relation $\delta$ 
(of a simple~P~module) 
supports duplex communication channels between adjacent P~system cells.  
Substantial modifications are needed when we consider the simplex communication case.  
It is not just a simple matter of changing the rules of the systems to 
only following out-neighbors when we are finding paths from the source 
to the target---we have explicitly utilized the ability to ``push back'' flow 
on a flow-arc, by sending objects to their flow-predecessors, 
when hunting for augmenting paths.  
Thus, some new ideas are needed before we can compute disjoint paths 
when the structural arcs allow only simplex communication.

We also want to know if we can we solve the problem of finding disjoint paths 
between $k$ pairs of $(s_1,t_1) \ldots (s_k,t_k)$, 
that is comparable in performance to the $O(n^3)$ algorithm of \cite{RobertsonS1995}.

We are interested to know whether a breadth-first search (BFS) approach
would be more beneficial than a depth-first search.  By using BFS we could
potentially exploit more of the parallel nature of P~systems. 

By combining this paper's results with our previous P~solutions
for the Byzantine problem~\cite{DKN-JLAP2010,DKN-CMC2010}, we have now solved one of our original goals, 
i.e. to solve, where possible, the Byzantine problem for P~systems
based on general digraphs, not necessarily complete. 
This more general problem can be solved in two phases:
(1) determine all node-disjoint paths in a digraph, 
assuming that, in this phase, there are no faults; and then 
(2) solve the consensus problem, even if, in this phase, 
some nodes fail in arbitrary, Byzantine ways.
An interesting problem arises, which, apparently, hasn't been considered yet.
What can we do if the Byzantine nodes already behave in a Byzantine manner 
in phase (1), while we attempt to build the node-disjoint paths?
Can we still determine all or a sufficient number of node-disjoint paths,
in the presence of Byzantine faults?

Some of previous experiences~\cite{DKN-MeCBIC2009,NDK-BWMC2009} have suggested that 
P~systems need to be extended with support for mobile arcs,
to incrementally build direct communication channels 
between originally distant cells and we have offered a preliminary solution.
The current experience suggests that this support should
be extended to enable straightforward creation of virtual 
search digraphs on top of existing physical digraphs.

We consider that this continued experience will provide good feedback on
the usability of P~systems as a formal model for parallel and distributed computing
and suggest a range of extensions and improvements, both at the conceptual level
and for practical implementations.


\section*{Acknowledgments}

The authors wish to thank Koray Altag, Masoud Khosravani, Huiling Wu 
and the three anonymous reviewers 
for valuable comments and feedback that helped us improve the paper.


\bibliographystyle{eptcs}
\bibliography{MyRef}  

\begin{thebibliography}{10}
\providecommand{\bibitemstart}[1]{\bibitem{#1}}
\providecommand{\bibitemend}{}
\providecommand{\bibliographystart}{}
\providecommand{\bibliographyend}{}
\providecommand{\url}[1]{\texttt{#1}}
\providecommand{\urlprefix}{Available at }
\providecommand{\bibinfo}[2]{#2}
\bibliographystart

\bibitemstart{DKN-MeCBIC2009}
\bibinfo{author}{Michael~J. Dinneen}, \bibinfo{author}{Yun-Bum Kim} \&
  \bibinfo{author}{Radu Nicolescu} (\bibinfo{year}{2009}):
  \emph{\bibinfo{title}{New Solutions to the Firing Squad Synchronization
  Problem for Neural and Hyperdag {P}~Systems}}.
\newblock In: {\sl \bibinfo{booktitle}{Membrane Computing and Biologically
  Inspired Process Calculi, Third Workshop, MeCBIC 2009, Bologna, Italy,
  September 5, 2009}}, pp. \bibinfo{pages}{117--130}.
\bibitemend

\bibitemstart{DKN-CMC2010}
\bibinfo{author}{Michael~J. Dinneen}, \bibinfo{author}{Yun-Bum Kim} \&
  \bibinfo{author}{Radu Nicolescu} (\bibinfo{year}{2010}):
  \emph{\bibinfo{title}{A Faster {P}~Solution for the {B}yzantine Agreement
  Problem}}.
\newblock In: {\sl \bibinfo{booktitle}{Eleventh International Conference on
  Membrane Computing, CMC11, Jena, Germany, August 24-27, 2010}}, p.
  \bibinfo{pages}{26pp}.
\bibitemend

\bibitemstart{DKN-JLAP2010}
\bibinfo{author}{Michael~J. Dinneen}, \bibinfo{author}{Yun-Bum Kim} \&
  \bibinfo{author}{Radu Nicolescu} (\bibinfo{year}{2010}):
  \emph{\bibinfo{title}{{P} systems and the {B}yzantine agreement}}.
\newblock {\sl \bibinfo{journal}{Journal of Logic and Algebraic Programming}}
  \bibinfo{volume}{79}(\bibinfo{number}{6}), pp. \bibinfo{pages}{334 -- 349}.
\newblock
  \urlprefix\url{http://www.sciencedirect.com/science/article/B6W8D-4YPPPW1-2/%
2/17b82b2cdd8f159b7fea380939193e4d}.
\newblock \bibinfo{note}{Membrane computing and programming}.
\bibitemend

\bibitemstart{EdmondsK1972}
\bibinfo{author}{Jack Edmonds} \& \bibinfo{author}{Richard~M. Karp}
  (\bibinfo{year}{1972}): \emph{\bibinfo{title}{Theoretical Improvements in
  Algorithmic Efficiency for Network Flow Problems}}.
\newblock {\sl \bibinfo{journal}{J. ACM}}
  \bibinfo{volume}{19}(\bibinfo{number}{2}), pp. \bibinfo{pages}{248--264}.
\newblock \urlprefix\url{http://doi.acm.org/10.1145/321694.321699}.
\bibitemend

\bibitemstart{FordF1956}
\bibinfo{author}{Lester R.~Ford Jr.} \& \bibinfo{author}{D.~Ray Fulkerson}
  (\bibinfo{year}{1956}): \emph{\bibinfo{title}{Maximal flow through a
  network}}.
\newblock {\sl \bibinfo{journal}{Canadian Journal of Mathematics}}
  \bibinfo{volume}{8}, pp. \bibinfo{pages}{399--404}.
\bibitemend

\bibitemstart{Lynch1996}
\bibinfo{author}{Nancy~A. Lynch} (\bibinfo{year}{1996}):
  \emph{\bibinfo{title}{Distributed Algorithms}}.
\newblock \bibinfo{publisher}{Morgan Kaufmann Publishers Inc.},
  \bibinfo{address}{San Francisco, CA, USA}.
\bibitemend

\bibitemstart{MartinVidePPR2003}
\bibinfo{author}{Carlos Mart\'{\i}n-Vide}, \bibinfo{author}{Gheorghe
  P{\u{a}}un}, \bibinfo{author}{Juan Pazos} \& \bibinfo{author}{Alfonso
  Rodr\'{\i}guez-Pat{\'o}n} (\bibinfo{year}{2003}):
  \emph{\bibinfo{title}{Tissue {P}~systems}}.
\newblock {\sl \bibinfo{journal}{Theor. Comput. Sci.}}
  \bibinfo{volume}{296}(\bibinfo{number}{2}), pp. \bibinfo{pages}{295--326}.
\newblock \urlprefix\url{http://dx.doi.org/10.1016/S0304-3975(02)00659-X}.
\bibitemend

\bibitemstart{Menger1927}
\bibinfo{author}{Karl Menger} (\bibinfo{year}{1927}): \emph{\bibinfo{title}{Zur
  allgemeinen Kurventheorie}}.
\newblock {\sl \bibinfo{journal}{Fundamenta Mathematicae}}
  \bibinfo{volume}{10}, pp. \bibinfo{pages}{96--115}.
\bibitemend

\bibitemstart{NDK-WMC2009}
\bibinfo{author}{Radu Nicolescu}, \bibinfo{author}{Michael~J. Dinneen} \&
  \bibinfo{author}{Yun-Bum Kim} (\bibinfo{year}{2009}):
  \emph{\bibinfo{title}{Discovering the Membrane Topology of Hyperdag
  {P}~Systems}}.
\newblock In: \bibinfo{editor}{Gheorghe P{\u{a}}un}, \bibinfo{editor}{Mario~J.
  P{\'e}rez-Jim{\'e}nez}, \bibinfo{editor}{Agust{\'{i}}n Riscos-N{\'u}{\~n}ez},
  \bibinfo{editor}{Grzegorz Rozenberg} \& \bibinfo{editor}{Arto Salomaa},
  editors: {\sl \bibinfo{booktitle}{Workshop on Membrane Computing}}, {\sl
  \bibinfo{series}{Lecture Notes in Computer Science}} \bibinfo{volume}{5957},
  \bibinfo{publisher}{Springer-Verlag}, pp. \bibinfo{pages}{410--435}.
\bibitemend

\bibitemstart{NDK-BWMC2009}
\bibinfo{author}{Radu Nicolescu}, \bibinfo{author}{Michael~J. Dinneen} \&
  \bibinfo{author}{Yun-Bum Kim} (\bibinfo{year}{2009}):
  \emph{\bibinfo{title}{Structured Modelling with Hyperdag {P}~Systems: Part
  {A}}}.
\newblock In: \bibinfo{editor}{Rosa Guti\'errez-Escudero},
  \bibinfo{editor}{Miguel~A. Guti{\'{e}}rrez-Naranjo},
  \bibinfo{editor}{Gheorghe P{\u{a}}un} \& \bibinfo{editor}{Ignacio
  P\'erez-Hurtado}, editors: {\sl \bibinfo{booktitle}{Brainstorming Week on
  Membrane Computing}}, ~\bibinfo{volume}{2}, \bibinfo{publisher}{Universidad
  de Sevilla}, pp. \bibinfo{pages}{85--107}.
\bibitemend

\bibitemstart{NDK-IJCCC2010}
\bibinfo{author}{Radu Nicolescu}, \bibinfo{author}{Michael~J. Dinneen} \&
  \bibinfo{author}{Yun-Bum Kim} (\bibinfo{year}{2010}):
  \emph{\bibinfo{title}{Towards Structured Modelling with Hyperdag
  {P}~Systems}}.
\newblock {\sl \bibinfo{journal}{International Journal of Computers,
  Communications and Control}} \bibinfo{volume}{2}, pp.
  \bibinfo{pages}{209--222}.
\bibitemend

\bibitemstart{Paun1998}
\bibinfo{author}{Gheorghe P{\u{a}}un} (\bibinfo{year}{1998}):
  \emph{\bibinfo{title}{Computing with Membranes}}.
\newblock \bibinfo{type}{Technical Report} \bibinfo{number}{208},
  \bibinfo{institution}{Turku Center for Computer Science-TUCS}.
\newblock \bibinfo{note}{(www.tucs.fi)}.
\bibitemend

\bibitemstart{Paun2002}
\bibinfo{author}{Gheorghe P{\u{a}}un} (\bibinfo{year}{2002}):
  \emph{\bibinfo{title}{Membrane Computing: An Introduction}}.
\newblock \bibinfo{publisher}{Springer-Verlag New York, Inc.},
  \bibinfo{address}{Secaucus, NJ, USA}.
\bibitemend

\bibitemstart{Paun2006}
\bibinfo{author}{Gheorghe P{\u{a}}un} (\bibinfo{year}{2006}):
  \emph{\bibinfo{title}{Introduction to Membrane Computing}}.
\newblock In: \bibinfo{editor}{Gabriel Ciobanu}, \bibinfo{editor}{Mario~J.
  P{\'e}rez-Jim{\'e}nez} \& \bibinfo{editor}{Gheorghe P{\u{a}}un}, editors:
  {\sl \bibinfo{booktitle}{Applications of Membrane Computing}},
  \bibinfo{series}{Natural Computing Series}, \bibinfo{publisher}{Springer},
  pp. \bibinfo{pages}{1--42}.
\newblock \urlprefix\url{http://dx.doi.org/10.1007/3-540-29937-8_1}.
\bibitemend

\bibitemstart{Robacker56}
\bibinfo{author}{John~T. Robacker} (\bibinfo{year}{1956}):
  \emph{\bibinfo{title}{Min-Max Theorems on Shortest Chains and Disjoint Cuts
  of a Network}}.
\newblock \bibinfo{type}{Research Memorandum RM-1660},
  \bibinfo{institution}{The RAND Corporation, Santa Monica, California}.
\bibitemend

\bibitemstart{RobertsonS1995}
\bibinfo{author}{Neil Robertson} \& \bibinfo{author}{Paul~D. Seymour}
  (\bibinfo{year}{1995}): \emph{\bibinfo{title}{Graph minors. {XIII}: {T}he
  disjoint paths problem}}.
\newblock {\sl \bibinfo{journal}{J. Comb. Theory Ser. B}}
  \bibinfo{volume}{63}(\bibinfo{number}{1}), pp. \bibinfo{pages}{65--110}.
\bibitemend

\bibitemstart{SeoT-IPDPS2009}
\bibinfo{author}{Daeho Seo} \& \bibinfo{author}{Mithuna Thottethodi}
  (\bibinfo{year}{2009}): \emph{\bibinfo{title}{Disjoint-path routing:
  Efficient communication for streaming applications}}.
\newblock In: {\sl \bibinfo{booktitle}{IPDPS}}, \bibinfo{publisher}{IEEE}, pp.
  \bibinfo{pages}{1--12}.
\newblock \urlprefix\url{http://dx.doi.org/10.1109/IPDPS.2009.5161081}.
\bibitemend

\bibliographyend
\end{thebibliography}


\begin{table}[h]
\caption{Edge-disjoint paths solution traces (steps $0, 1, \ldots, 30$) of the 
simple~P~module shown in Figure~\ref{fig-virtual-search-digraph}~(a),
where $\sigma_1$ is the source cell and $\sigma_6$ is the target cell.}
\label{tab-trace}
\begin{center}
\renewcommand{\tabcolsep}{3.0pt}
\renewcommand{\arraystretch}{1.3}
\footnotesize
\noindent
\begin{tabular}{ | l | l | l | l | l | l | l | }
\hline
Step$\backslash$Cell
 & $\sigma_1$ & $\sigma_2$ & $\sigma_3$ & $\sigma_4$ & $\sigma_5$ & $\sigma_6$ \\ \hline
0 & $s_0~ g_6$ & $s_0~ $ & $s_0~ $ & $s_0~ $ & $s_0~ $ & $s_0~ $ \\ \hline
1 & $s_1~ ak$ & $s_0~ u_6$ & $s_0~ $ & $s_0~ u_6$ & $s_0~ $ & $s_0~ $ \\ \hline
2 & $s_2~ aku_6^{2}$ & $s_1~ an_1u_6$ & $s_0~ u_6^{2}$ & $s_1~ an_1u_6$ & $s_0~ u_6$ & $s_0~ $ \\ \hline
3 & $s_3~ akn_2n_4u_6^{2}$ & $s_2~ an_1n_4u_6^{2}$ & $s_1~ an_2n_4u_6^{2}$ & $s_2~ an_1n_2u_6^{3}$ & $s_1~ an_4u_6$ & $s_0~ u_6^{2}$ \\ \hline
4 & \myway{$s_4~ akn_2n_4$} & $s_3~ an_1n_3n_4u_6^{2}$ & $s_2~ an_2n_4n_5u_6^{3}$ & $s_3~ an_1n_2n_3n_5u_6^{3}$ & $s_2~ an_3n_4u_6^{2}$ & $s_1~ an_3n_5z$ \\ \hline
5 & $s_5~ akn_2n_4$ & \myway{$s_4~ an_1n_3n_4$} & $s_3~ an_2n_4n_5n_6u_6^{3}$ & \myway{$s_4~ an_1n_2n_3n_5$} & $s_3~ an_3n_4n_6u_6^{2}$ & $s_2~ an_3n_5z$ \\ \hline
6 & $s_5~ d_4kn_2$ & $s_7~ an_1n_3n_4$ & \myway{$s_4~ an_2n_4n_5n_6$} & $s_7~ af_1n_1n_2n_3n_5$ & \myway{$s_4~ an_3n_4n_6$} & $s_3~ an_3n_5z$ \\ \hline
7 & $s_5~ d_4kn_2$ & $s_7~ an_1n_3n_4$ & $s_7~ an_2n_4n_5n_6$ & $s_8~ an_2n_3n_5q_1$ & $s_7~ an_3n_4n_6$ & \myway{$s_4~ an_3n_5z$} \\ \hline
8 & $s_5~ d_4kn_2$ & $s_7~ an_1n_3n_4$ & $s_7~ af_4n_2n_4n_5n_6$ & $s_9~ ad_3n_2n_5q_1$ & $s_7~ an_3n_4n_6$ & $s_6~ an_3n_5z$ \\ \hline
9 & $s_5~ d_4kn_2$ & $s_7~ an_1n_3n_4$ & $s_8~ an_2n_5n_6q_4$ & $s_9~ ad_3n_2n_5q_1$ & $s_7~ an_3n_4n_6$ & $s_6~ an_3n_5z$ \\ \hline
10 & $s_5~ d_4kn_2$ & $s_7~ af_3n_1n_3n_4$ & $s_9~ ad_2n_5n_6q_4$ & $s_9~ ad_3n_2n_5q_1$ & $s_7~ an_3n_4n_6$ & $s_6~ an_3n_5z$ \\ \hline
11 & $s_5~ d_4kn_2$ & $s_8~ an_1n_4q_3$ & $s_9~ ad_2n_5n_6q_4$ & $s_9~ ad_3n_2n_5q_1$ & $s_7~ an_3n_4n_6$ & $s_6~ an_3n_5z$ \\ \hline
12 & $s_5~ d_4f_2kn_2$ & $s_9~ ad_1n_4q_3$ & $s_9~ ad_2n_5n_6q_4$ & $s_9~ ad_3n_2n_5q_1$ & $s_7~ an_3n_4n_6$ & $s_6~ an_3n_5z$ \\ \hline
13 & $s_5~ d_4kn_2$ & $s_9~ ad_1n_4q_3x_1$ & $s_9~ ad_2n_5n_6q_4$ & $s_9~ ad_3n_2n_5q_1$ & $s_7~ an_3n_4n_6$ & $s_6~ an_3n_5z$ \\ \hline
14 & $s_5~ d_4kn_2$ & $s_8~ am_1n_4q_3$ & $s_9~ ad_2n_5n_6q_4$ & $s_9~ ad_3n_2n_5q_1$ & $s_7~ an_3n_4n_6$ & $s_6~ an_3n_5z$ \\ \hline
15 & $s_5~ d_4kn_2$ & $s_9~ ad_4m_1q_3$ & $s_9~ ad_2n_5n_6q_4$ & $s_9~ ad_3f_2n_2n_5q_1$ & $s_7~ an_3n_4n_6$ & $s_6~ an_3n_5z$ \\ \hline
16 & $s_5~ d_4kn_2$ & $s_9~ ad_4m_1q_3x_4$ & $s_9~ ad_2n_5n_6q_4$ & $s_9~ ad_3m_2n_5q_1$ & $s_7~ an_3n_4n_6$ & $s_6~ an_3n_5z$ \\ \hline
17 & $s_5~ d_4kn_2$ & $s_8~ am_1m_4q_3$ & $s_9~ ad_2n_5n_6q_4$ & $s_9~ ad_3m_2n_5q_1$ & $s_7~ an_3n_4n_6$ & $s_6~ an_3n_5z$ \\ \hline
18 & $s_5~ d_4kn_2$ & $s_{10}~ am_1m_4q_3$ & $s_9~ ad_2n_5n_6q_4$ & $s_9~ ad_3m_2n_5q_1$ & $s_7~ an_3n_4n_6$ & $s_6~ an_3n_5z$ \\ \hline
19 & $s_5~ d_4kn_2$ & $s_7~ am_1m_3m_4$ & $s_9~ ad_2n_5n_6q_4x_2$ & $s_9~ ad_3m_2n_5q_1$ & $s_7~ an_3n_4n_6$ & $s_6~ an_3n_5z$ \\ \hline
20 & $s_5~ d_4kn_2$ & $s_7~ am_1m_3m_4$ & $s_8~ am_2n_5n_6q_4$ & $s_9~ ad_3m_2n_5q_1$ & $s_7~ an_3n_4n_6$ & $s_6~ an_3n_5z$ \\ \hline
21 & $s_5~ d_4kn_2$ & $s_7~ am_1m_3m_4$ & $s_9~ ad_6m_2n_5q_4$ & $s_9~ ad_3m_2n_5q_1$ & $s_7~ an_3n_4n_6$ & $s_6~ af_3n_3n_5z$ \\ \hline
22 & $s_5~ d_4kn_2$ & $s_7~ am_1m_3m_4$ & $s_9~ ad_6m_2n_5q_4y_6$ & $s_9~ ad_3m_2n_5q_1$ & $s_7~ an_3n_4n_6$ & \myway{$s_6~ an_5p_3z$} \\ \hline
23 & $s_5~ d_4kn_2$ & $s_7~ am_1m_3m_4$ & \myway{$s_7~ ac_6m_2n_5p_4$} & $s_9~ ad_3m_2n_5q_1y_3$ & $s_7~ an_3n_4n_6$ & $s_6~ an_5p_3z$ \\ \hline
24 & $s_5~ d_4kn_2y_4$ & $s_7~ am_1m_3m_4$ & $s_7~ ac_6m_2n_5p_4$ & \myway{$s_7~ ac_3m_2n_5p_1$} & $s_7~ an_3n_4n_6$ & $s_6~ an_5p_3z$ \\ \hline
25 & \myway{$s_{12}~ ac_4kn_2w^{2}$} & $s_7~ am_1m_3m_4v$ & $s_7~ ac_6m_2n_5p_4$ & $s_7~ ac_3m_2n_5p_1v$ & $s_7~ an_3n_4n_6$ & $s_6~ an_5p_3z$ \\ \hline
26 & $s_{12}~ ac_4kn_2v^{2}w$ & $s_{12}~ am_1m_3m_4vw^{2}$ & $s_7~ ac_6m_2n_5p_4v^{2}$ & $s_{12}~ ac_3m_2n_5p_1vw^{2}$ & $s_7~ an_3n_4n_6v$ & $s_6~ an_5p_3z$ \\ \hline
27 & $s_{12}~ ac_4kn_2$ & $s_{12}~ an_1n_3n_4vw$ & $s_{12}~ ac_6m_2n_5p_4v^{2}w^{2}$ & $s_{12}~ ac_3n_2n_5p_1v^{2}w$ & $s_{12}~ an_3n_4n_6vw^{2}$ & $s_6~ an_5p_3v^{2}z$ \\ \hline
28 & $s_5~ ac_4kn_2$ & $s_{12}~ an_1n_3n_4$ & $s_{12}~ ac_6n_2n_5p_4vw$ & $s_{12}~ ac_3n_2n_5p_1$ & $s_{12}~ an_3n_4n_6vw$ & $s_{12}~ an_5p_3vw^{2}z$ \\ \hline
29 & $s_5~ c_4d_2k$ & $s_7~ af_1n_1n_3n_4$ & $s_{12}~ ac_6n_2n_5p_4$ & $s_7~ ac_3n_2n_5p_1$ & $s_{12}~ an_3n_4n_6$ & $s_{12}~ an_5p_3wz$ \\ \hline
30 & $s_5~ c_4d_2k$ & $s_8~ an_3n_4q_1$ & $s_7~ ac_6n_2n_5p_4$ & $s_7~ ac_3n_2n_5p_1$ & $s_7~ an_3n_4n_6$ & $s_{12}~ an_5p_3z$ \\ \hline
\end{tabular}
\end{center}
\end{table}

\begin{table}[h]
\caption{Edge-disjoint paths solution traces (steps $31, 32, \ldots, 61$) of the 
simple~P~module shown in Figure~\ref{fig-virtual-search-digraph}~(a),
where $\sigma_1$ is the source cell and $\sigma_6$ is the target cell.}
\label{tab-trace1}
\begin{center}
\renewcommand{\tabcolsep}{3.0pt}
\renewcommand{\arraystretch}{1.3}
\footnotesize
\noindent
\begin{tabular}{ | l | l | l | l | l | l | l | }
\hline
Step$\backslash$Cell
 & $\sigma_1$ & $\sigma_2$ & $\sigma_3$ & $\sigma_4$ & $\sigma_5$ & $\sigma_6$ \\ \hline
31 & $s_5~ c_4d_2k$ & $s_9~ ad_3n_4q_1$ & $s_7~ ac_6f_2n_2n_5p_4$ & $s_7~ ac_3n_2n_5p_1$ & $s_7~ an_3n_4n_6$ & $s_6~ an_5p_3z$ \\ \hline
32 & $s_5~ c_4d_2k$ & $s_9~ ad_3n_4q_1$ & $s_8~ ac_6n_5p_4q_2$ & $s_7~ ac_3n_2n_5p_1$ & $s_7~ an_3n_4n_6$ & $s_6~ an_5p_3z$ \\ \hline
33 & $s_5~ c_4d_2k$ & $s_9~ ad_3n_4q_1$ & $s_9~ ac_6d_5p_4q_2$ & $s_7~ ac_3n_2n_5p_1$ & $s_7~ af_3n_3n_4n_6$ & $s_6~ an_5p_3z$ \\ \hline
34 & $s_5~ c_4d_2k$ & $s_9~ ad_3n_4q_1$ & $s_9~ ac_6d_5p_4q_2$ & $s_7~ ac_3n_2n_5p_1$ & $s_8~ an_4n_6q_3$ & $s_6~ an_5p_3z$ \\ \hline
35 & $s_5~ c_4d_2k$ & $s_9~ ad_3n_4q_1$ & $s_9~ ac_6d_5p_4q_2$ & $s_7~ ac_3f_5n_2n_5p_1$ & $s_9~ ad_4n_6q_3$ & $s_6~ an_5p_3z$ \\ \hline
36 & $s_5~ c_4d_2k$ & $s_9~ ad_3n_4q_1$ & $s_9~ ac_6d_5p_4q_2$ & $s_8~ ac_3n_2p_1q_5$ & $s_9~ ad_4n_6q_3$ & $s_6~ an_5p_3z$ \\ \hline
37 & $s_5~ c_4d_2k$ & $s_9~ ad_3f_4n_4q_1$ & $s_9~ ac_6d_5p_4q_2$ & $s_9~ ac_3d_2p_1q_5$ & $s_9~ ad_4n_6q_3$ & $s_6~ an_5p_3z$ \\ \hline
38 & $s_5~ c_4d_2k$ & $s_9~ ad_3m_4q_1$ & $s_9~ ac_6d_5p_4q_2$ & $s_9~ ac_3d_2p_1q_5x_2$ & $s_9~ ad_4n_6q_3$ & $s_6~ an_5p_3z$ \\ \hline
39 & $s_5~ c_4d_2k$ & $s_9~ ad_3m_4q_1$ & $s_9~ ac_6d_5p_4q_2$ & $s_8~ ac_3m_2p_1q_5$ & $s_9~ ad_4n_6q_3$ & $s_6~ an_5p_3z$ \\ \hline
40 & $s_5~ c_4d_2k$ & $s_9~ ad_3m_4q_1$ & $s_9~ ac_6d_5p_4q_2$ & $s_{10}~ ac_3m_2p_1q_5$ & $s_9~ ad_4n_6q_3$ & $s_6~ an_5p_3z$ \\ \hline
41 & $s_5~ b_4c_4d_2k$ & $s_9~ ad_3m_4q_1$ & $s_9~ ac_6d_5p_4q_2$ & $s_{11}~ ac_3m_2q_5r_1$ & $s_9~ ad_4n_6q_3$ & $s_6~ an_5p_3z$ \\ \hline
42 & $s_5~ c_4d_2k$ & $s_9~ ad_3m_4q_1$ & $s_9~ ac_6d_5p_4q_2$ & $s_{11}~ ac_3m_2q_5r_1x_1$ & $s_9~ ad_4n_6q_3$ & $s_6~ an_5p_3z$ \\ \hline
43 & $s_5~ c_4d_2k$ & $s_9~ ad_3m_4q_1$ & $s_9~ ac_6d_5p_4q_2$ & $s_{10}~ ac_3m_2q_5t_1$ & $s_9~ ad_4n_6q_3$ & $s_6~ an_5p_3z$ \\ \hline
44 & $s_5~ c_4d_2k$ & $s_9~ ad_3m_4q_1$ & $s_9~ ac_6d_5p_4q_2$ & $s_7~ ac_3m_2m_5t_1$ & $s_9~ ad_4n_6q_3x_4$ & $s_6~ an_5p_3z$ \\ \hline
45 & $s_5~ c_4d_2k$ & $s_9~ ad_3m_4q_1$ & $s_9~ ac_6d_5p_4q_2$ & $s_7~ ac_3m_2m_5t_1$ & $s_8~ am_4n_6q_3$ & $s_6~ an_5p_3z$ \\ \hline
46 & $s_5~ c_4d_2k$ & $s_9~ ad_3m_4q_1$ & $s_9~ ac_6d_5p_4q_2$ & $s_7~ ac_3m_2m_5t_1$ & $s_9~ ad_6m_4q_3$ & $s_6~ af_5n_5p_3z$ \\ \hline
47 & $s_5~ c_4d_2k$ & $s_9~ ad_3m_4q_1$ & $s_9~ ac_6d_5p_4q_2$ & $s_7~ ac_3m_2m_5t_1$ & $s_9~ ad_6m_4q_3y_6$ & \myway{$s_6~ ap_3p_5z$} \\ \hline
48 & $s_5~ c_4d_2k$ & $s_9~ ad_3m_4q_1$ & $s_9~ ac_6d_5p_4q_2y_5$ & $s_7~ ac_3m_2m_5t_1$ & \myway{$s_7~ ac_6m_4p_3$} & $s_6~ ap_3p_5z$ \\ \hline
49 & $s_5~ c_4d_2k$ & $s_9~ ad_3m_4q_1y_3$ & \myway{$s_7~ ac_5c_6p_2p_4$} & $s_7~ ac_3m_2m_5t_1$ & $s_7~ ac_6m_4p_3$ & $s_6~ ap_3p_5z$ \\ \hline
50 & $s_5~ c_4d_2ky_2$ & \myway{$s_7~ ac_3m_4p_1$} & $s_7~ ac_5c_6p_2p_4$ & $s_7~ ac_3m_2m_5t_1$ & $s_7~ ac_6m_4p_3$ & $s_6~ ap_3p_5z$ \\ \hline
51 & \myway{$s_{12}~ ac_2c_4kw^{2}$} & $s_7~ ac_3m_4p_1v$ & $s_7~ ac_5c_6p_2p_4$ & $s_7~ ac_3m_2m_5t_1v$ & $s_7~ ac_6m_4p_3$ & $s_6~ ap_3p_5z$ \\ \hline
52 & $s_{12}~ ac_2c_4kv^{2}w$ & $s_{12}~ ac_3m_4p_1vw^{2}$ & $s_7~ ac_5c_6p_2p_4v^{2}$ & $s_{12}~ ac_3m_2m_5t_1vw^{2}$ & $s_7~ ac_6m_4p_3v$ & $s_6~ ap_3p_5z$ \\ \hline
53 & $s_{12}~ ac_2c_4k$ & $s_{12}~ ac_3n_4p_1vw$ & $s_{12}~ ac_5c_6p_2p_4v^{2}w^{2}$ & $s_{12}~ ac_3n_2n_5p_1v^{2}w$ & $s_{12}~ ac_6m_4p_3vw^{2}$ & $s_6~ ap_3p_5v^{2}z$ \\ \hline
54 & $s_5~ ac_2c_4k$ & $s_{12}~ ac_3n_4p_1$ & $s_{12}~ ac_5c_6p_2p_4vw$ & $s_{12}~ ac_3n_2n_5p_1$ & $s_{12}~ ac_6n_4p_3vw$ & $s_{12}~ ap_3p_5vw^{2}z$ \\ \hline
55 & $s_{13}~ a^{2}c_2c_4w^{2}$ & $s_7~ a^{2}c_3n_4p_1$ & $s_{12}~ ac_5c_6p_2p_4$ & $s_7~ a^{2}c_3n_2n_5p_1$ & $s_{12}~ ac_6n_4p_3$ & $s_{12}~ ap_3p_5wz$ \\ \hline
56 & $s_{13}~ a^{4}c_2c_4w$ & $s_{13}~ a^{3}c_3n_4p_1w^{2}$ & $s_7~ a^{3}c_5c_6p_2p_4$ & $s_{13}~ a^{3}c_3n_2n_5p_1w^{2}$ & $s_7~ a^{2}c_6n_4p_3$ & $s_{12}~ ap_3p_5z$ \\ \hline
57 & $s_{13}~ a^{4}c_2c_4$ & $s_{13}~ a^{4}c_3n_4p_1w$ & $s_{13}~ a^{4}c_5c_6p_2p_4w^{2}$ & $s_{13}~ a^{5}c_3n_2n_5p_1w$ & $s_{13}~ a^{3}c_6n_4p_3w^{2}$ & $s_6~ a^{3}p_3p_5z$ \\ \hline
58 & $s_0~ c_2c_4$ & $s_{13}~ a^{4}c_3n_4p_1$ & $s_{13}~ a^{5}c_5c_6p_2p_4w$ & $s_{13}~ a^{5}c_3n_2n_5p_1$ & $s_{13}~ a^{4}c_6n_4p_3w$ & $s_{13}~ a^{3}p_3p_5w^{2}$ \\ \hline
59 & $s_0~ c_2c_4$ & $s_0~ c_3p_1$ & $s_{13}~ a^{5}c_5c_6p_2p_4$ & $s_0~ c_3p_1$ & $s_{13}~ a^{4}c_6n_4p_3$ & $s_{13}~ a^{3}p_3p_5w$ \\ \hline
60 & $s_0~ c_2c_4$ & $s_0~ c_3p_1$ & $s_0~ c_5c_6p_2p_4$ & $s_0~ c_3p_1$ & $s_0~ c_6p_3$ & $s_{13}~ a^{3}p_3p_5$ \\ \hline
61 & \myway{$s_0~ c_2c_4$} & \myway{$s_0~ c_3p_1$} & \myway{$s_0~ c_5c_6p_2p_4$} & \myway{$s_0~ c_3p_1$} & \myway{$s_0~ c_6p_3$} & \myway{$s_0~ p_3p_5$} \\ \hline
\end{tabular}
\end{center}
\end{table}


\begin{table}[h]
\caption{Node-disjoint paths solution traces (steps $0, 1, \ldots, 29$) of the 
simple~P~module shown in Figure~\ref{fig-virtual-search-digraph}~(a),
where $\sigma_1$ is the source cell and $\sigma_6$ is the target cell.}
\label{tab-trace-node}
\begin{center}
\renewcommand{\tabcolsep}{3.0pt}
\renewcommand{\arraystretch}{1.3}
\footnotesize
\noindent
\begin{tabular}{ | l | l | l | l | l | l | l | }
\hline
Step$\backslash$Cell
 & $\sigma_1$ & $\sigma_2$ & $\sigma_3$ & $\sigma_4$ & $\sigma_5$ & $\sigma_6$ \\ \hline
0 & $s_0~ g_6$ & $s_0~ $ & $s_0~ $ & $s_0~ $ & $s_0~ $ & $s_0~ $ \\ \hline
1 & $s_1~ ak$ & $s_0~ u_6$ & $s_0~ $ & $s_0~ u_6$ & $s_0~ $ & $s_0~ $ \\ \hline
2 & $s_2~ aku_6^{2}$ & $s_1~ an_1u_6$ & $s_0~ u_6^{2}$ & $s_1~ an_1u_6$ & $s_0~ u_6$ & $s_0~ $ \\ \hline
3 & $s_3~ akn_2n_4u_6^{2}$ & $s_2~ an_1n_4u_6^{2}$ & $s_1~ an_2n_4u_6^{2}$ & $s_2~ an_1n_2u_6^{3}$ & $s_1~ an_4u_6$ & $s_0~ u_6^{2}$ \\ \hline
4 & \myway{$s_4~ akn_2n_4$} & $s_3~ an_1n_3n_4u_6^{2}$ & $s_2~ an_2n_4n_5u_6^{3}$ & $s_3~ an_1n_2n_3n_5u_6^{3}$ & $s_2~ an_3n_4u_6^{2}$ & $s_1~ an_3n_5z$ \\ \hline
5 & $s_5~ akn_2n_4$ & \myway{$s_4~ an_1n_3n_4$} & $s_3~ an_2n_4n_5n_6u_6^{3}$ & \myway{$s_4~ an_1n_2n_3n_5$} & $s_3~ an_3n_4n_6u_6^{2}$ & $s_2~ an_3n_5z$ \\ \hline
6 & $s_5~ d_4kn_2$ & $s_7~ an_1n_3n_4$ & \myway{$s_4~ an_2n_4n_5n_6$} & $s_7~ af_1n_1n_2n_3n_5$ & \myway{$s_4~ an_3n_4n_6$} & $s_3~ an_3n_5z$ \\ \hline
7 & $s_5~ d_4kn_2$ & $s_7~ an_1n_3n_4$ & $s_7~ an_2n_4n_5n_6$ & $s_8~ an_2n_3n_5q_1$ & $s_7~ an_3n_4n_6$ & \myway{$s_4~ an_3n_5z$} \\ \hline
8 & $s_5~ d_4kn_2$ & $s_7~ an_1n_3n_4$ & $s_7~ af_4n_2n_4n_5n_6$ & $s_9~ ad_3n_2n_5q_1$ & $s_7~ an_3n_4n_6$ & $s_6~ an_3n_5z$ \\ \hline
9 & $s_5~ d_4kn_2$ & $s_7~ an_1n_3n_4$ & $s_8~ an_2n_5n_6q_4$ & $s_9~ ad_3n_2n_5q_1$ & $s_7~ an_3n_4n_6$ & $s_6~ an_3n_5z$ \\ \hline
10 & $s_5~ d_4kn_2$ & $s_7~ af_3n_1n_3n_4$ & $s_9~ ad_2n_5n_6q_4$ & $s_9~ ad_3n_2n_5q_1$ & $s_7~ an_3n_4n_6$ & $s_6~ an_3n_5z$ \\ \hline
11 & $s_5~ d_4kn_2$ & $s_8~ an_1n_4q_3$ & $s_9~ ad_2n_5n_6q_4$ & $s_9~ ad_3n_2n_5q_1$ & $s_7~ an_3n_4n_6$ & $s_6~ an_3n_5z$ \\ \hline
12 & $s_5~ d_4f_2kn_2$ & $s_9~ ad_1n_4q_3$ & $s_9~ ad_2n_5n_6q_4$ & $s_9~ ad_3n_2n_5q_1$ & $s_7~ an_3n_4n_6$ & $s_6~ an_3n_5z$ \\ \hline
13 & $s_5~ d_4kn_2$ & $s_9~ ad_1n_4q_3x_1$ & $s_9~ ad_2n_5n_6q_4$ & $s_9~ ad_3n_2n_5q_1$ & $s_7~ an_3n_4n_6$ & $s_6~ an_3n_5z$ \\ \hline
14 & $s_5~ d_4kn_2$ & $s_8~ am_1n_4q_3$ & $s_9~ ad_2n_5n_6q_4$ & $s_9~ ad_3n_2n_5q_1$ & $s_7~ an_3n_4n_6$ & $s_6~ an_3n_5z$ \\ \hline
15 & $s_5~ d_4kn_2$ & $s_9~ ad_4m_1q_3$ & $s_9~ ad_2n_5n_6q_4$ & $s_9~ ad_3f_2n_2n_5q_1$ & $s_7~ an_3n_4n_6$ & $s_6~ an_3n_5z$ \\ \hline
16 & $s_5~ d_4kn_2$ & $s_9~ ad_4m_1q_3x_4$ & $s_9~ ad_2n_5n_6q_4$ & $s_9~ ad_3m_2n_5q_1$ & $s_7~ an_3n_4n_6$ & $s_6~ an_3n_5z$ \\ \hline
17 & $s_5~ d_4kn_2$ & $s_8~ am_1m_4q_3$ & $s_9~ ad_2n_5n_6q_4$ & $s_9~ ad_3m_2n_5q_1$ & $s_7~ an_3n_4n_6$ & $s_6~ an_3n_5z$ \\ \hline
18 & $s_5~ d_4kn_2$ & $s_{10}~ am_1m_4q_3$ & $s_9~ ad_2n_5n_6q_4$ & $s_9~ ad_3m_2n_5q_1$ & $s_7~ an_3n_4n_6$ & $s_6~ an_3n_5z$ \\ \hline
19 & $s_5~ d_4kn_2$ & $s_7~ am_1m_3m_4$ & $s_9~ ad_2n_5n_6q_4x_2$ & $s_9~ ad_3m_2n_5q_1$ & $s_7~ an_3n_4n_6$ & $s_6~ an_3n_5z$ \\ \hline
20 & $s_5~ d_4kn_2$ & $s_7~ am_1m_3m_4$ & $s_8~ am_2n_5n_6q_4$ & $s_9~ ad_3m_2n_5q_1$ & $s_7~ an_3n_4n_6$ & $s_6~ an_3n_5z$ \\ \hline
21 & $s_5~ d_4kn_2$ & $s_7~ am_1m_3m_4$ & $s_9~ ad_6m_2n_5q_4$ & $s_9~ ad_3m_2n_5q_1$ & $s_7~ an_3n_4n_6$ & $s_6~ af_3n_3n_5z$ \\ \hline
22 & $s_5~ d_4kn_2$ & $s_7~ am_1m_3m_4$ & $s_9~ ad_6m_2n_5q_4y_6$ & $s_9~ ad_3m_2n_5q_1$ & $s_7~ an_3n_4n_6$ & \myway{$s_6~ an_5p_3z$} \\ \hline
23 & $s_5~ d_4kn_2$ & $s_7~ am_1m_3m_4$ & \myway{$s_7~ ac_6m_2n_5p_4$} & $s_9~ ad_3m_2n_5q_1y_3$ & $s_7~ an_3n_4n_6$ & $s_6~ an_5p_3z$ \\ \hline
24 & $s_5~ d_4kn_2y_4$ & $s_7~ am_1m_3m_4$ & $s_7~ ac_6m_2n_5p_4$ & \myway{$s_7~ ac_3m_2n_5p_1$} & $s_7~ an_3n_4n_6$ & $s_6~ an_5p_3z$ \\ \hline
25 & \myway{$s_{12}~ ac_4kn_2w^{2}$} & $s_7~ am_1m_3m_4v$ & $s_7~ ac_6m_2n_5p_4$ & $s_7~ ac_3m_2n_5p_1v$ & $s_7~ an_3n_4n_6$ & $s_6~ an_5p_3z$ \\ \hline
26 & $s_{12}~ ac_4kn_2v^{2}w$ & $s_{12}~ am_1m_3m_4vw^{2}$ & $s_7~ ac_6m_2n_5p_4v^{2}$ & $s_{12}~ ac_3m_2n_5p_1vw^{2}$ & $s_7~ an_3n_4n_6v$ & $s_6~ an_5p_3z$ \\ \hline
27 & $s_{12}~ ac_4kn_2$ & $s_{12}~ an_1n_3n_4vw$ & $s_{12}~ ac_6m_2n_5p_4v^{2}w^{2}$ & $s_{12}~ ac_3n_2n_5p_1v^{2}w$ & $s_{12}~ an_3n_4n_6vw^{2}$ & $s_6~ an_5p_3v^{2}z$ \\ \hline
28 & $s_5~ ac_4kn_2$ & $s_{12}~ an_1n_3n_4$ & $s_{12}~ ac_6n_2n_5p_4vw$ & $s_{12}~ ac_3n_2n_5p_1$ & $s_{12}~ an_3n_4n_6vw$ & $s_{12}~ an_5p_3vw^{2}z$ \\ \hline
29 & $s_5~ c_4d_2k$ & $s_7~ af_1n_1n_3n_4$ & $s_{12}~ ac_6n_2n_5p_4$ & $s_7~ ac_3n_2n_5p_1$ & $s_{12}~ an_3n_4n_6$ & $s_{12}~ an_5p_3wz$ \\ \hline
\end{tabular}
\end{center}
\end{table}

\begin{table}[h]
\caption{node-disjoint paths solution traces (steps $30, 31, \ldots, 59$) of the 
simple~P~module shown in Figure~\ref{fig-virtual-search-digraph}~(a),
where $\sigma_1$ is the source cell and $\sigma_6$ is the target cell.}
\label{tab-trace1-node}
\begin{center}
\renewcommand{\tabcolsep}{3.0pt}
\renewcommand{\arraystretch}{1.3}
\footnotesize
\noindent
\begin{tabular}{ | l | l | l | l | l | l | l | }
\hline
Step$\backslash$Cell
 & $\sigma_1$ & $\sigma_2$ & $\sigma_3$ & $\sigma_4$ & $\sigma_5$ & $\sigma_6$ \\ \hline
30 & $s_5~ c_4d_2k$ & $s_8~ an_3n_4q_1$ & $s_7~ ac_6n_2n_5p_4$ & $s_7~ ac_3n_2n_5p_1$ & $s_7~ an_3n_4n_6$ & $s_{12}~ an_5p_3z$ \\ \hline
31 & $s_5~ c_4d_2k$ & $s_9~ ad_3n_4q_1$ & $s_7~ ac_6f_2n_2n_5p_4$ & $s_7~ ac_3n_2n_5p_1$ & $s_7~ an_3n_4n_6$ & $s_6~ an_5p_3z$ \\ \hline
32 & $s_5~ c_4d_2k$ & $s_9~ ad_3n_4q_1$ & $s_{10}~ ac_6n_5p_4q_2$ & $s_7~ ac_3n_2n_5p_1$ & $s_7~ an_3n_4n_6$ & $s_6~ an_5p_3z$ \\ \hline
33 & $s_5~ c_4d_2k$ & $s_9~ ad_3n_4q_1$ & $s_{11}~ ac_6n_5q_2r_4$ & $s_7~ ab_3c_3n_2n_5p_1$ & $s_7~ an_3n_4n_6$ & $s_6~ an_5p_3z$ \\ \hline
34 & $s_5~ c_4d_2k$ & $s_9~ ad_3n_4q_1$ & $s_{11}~ ac_6n_5q_2r_4$ & $s_8~ ae_3n_2n_5p_1$ & $s_7~ an_3n_4n_6$ & $s_6~ an_5p_3z$ \\ \hline
35 & $s_5~ c_4d_2k$ & $s_9~ ad_3f_4n_4q_1$ & $s_{11}~ ac_6n_5q_2r_4$ & $s_9~ ad_2e_3n_5p_1$ & $s_7~ an_3n_4n_6$ & $s_6~ an_5p_3z$ \\ \hline
36 & $s_5~ c_4d_2k$ & $s_9~ ad_3m_4q_1$ & $s_{11}~ ac_6n_5q_2r_4$ & $s_9~ ad_2e_3n_5p_1x_2$ & $s_7~ an_3n_4n_6$ & $s_6~ an_5p_3z$ \\ \hline
37 & $s_5~ c_4d_2k$ & $s_9~ ad_3m_4q_1$ & $s_{11}~ ac_6n_5q_2r_4$ & $s_8~ ae_3m_2n_5p_1$ & $s_7~ an_3n_4n_6$ & $s_6~ an_5p_3z$ \\ \hline
38 & $s_5~ c_4d_2k$ & $s_9~ ad_3m_4q_1$ & $s_{11}~ ac_6n_5q_2r_4$ & $s_9~ ad_5e_3m_2p_1$ & $s_7~ af_4n_3n_4n_6$ & $s_6~ an_5p_3z$ \\ \hline
39 & $s_5~ c_4d_2k$ & $s_9~ ad_3m_4q_1$ & $s_{11}~ ac_6n_5q_2r_4$ & $s_9~ ad_5e_3m_2p_1$ & $s_8~ an_3n_6q_4$ & $s_6~ an_5p_3z$ \\ \hline
40 & $s_5~ c_4d_2k$ & $s_9~ ad_3m_4q_1$ & $s_{11}~ ac_6f_5n_5q_2r_4$ & $s_9~ ad_5e_3m_2p_1$ & $s_9~ ad_3n_6q_4$ & $s_6~ an_5p_3z$ \\ \hline
41 & $s_5~ c_4d_2k$ & $s_9~ ad_3m_4q_1$ & $s_{11}~ ac_6m_5q_2r_4$ & $s_9~ ad_5e_3m_2p_1$ & $s_9~ ad_3n_6q_4x_3$ & $s_6~ an_5p_3z$ \\ \hline
42 & $s_5~ c_4d_2k$ & $s_9~ ad_3m_4q_1$ & $s_{11}~ ac_6m_5q_2r_4$ & $s_9~ ad_5e_3m_2p_1$ & $s_8~ am_3n_6q_4$ & $s_6~ an_5p_3z$ \\ \hline
43 & $s_5~ c_4d_2k$ & $s_9~ ad_3m_4q_1$ & $s_{11}~ ac_6m_5q_2r_4$ & $s_9~ ad_5e_3m_2p_1$ & $s_9~ ad_6m_3q_4$ & $s_6~ af_5n_5p_3z$ \\ \hline
44 & $s_5~ c_4d_2k$ & $s_9~ ad_3m_4q_1$ & $s_{11}~ ac_6m_5q_2r_4$ & $s_9~ ad_5e_3m_2p_1$ & $s_9~ ad_6m_3q_4y_6$ & \myway{$s_6~ ap_3p_5z$} \\ \hline
45 & $s_5~ c_4d_2k$ & $s_9~ ad_3m_4q_1$ & $s_{11}~ ac_6m_5q_2r_4$ & $s_9~ ad_5e_3m_2p_1y_5$ & \myway{$s_7~ ac_6m_3p_4$} & $s_6~ ap_3p_5z$ \\ \hline
46 & $s_5~ c_4d_2k$ & $s_9~ ad_3m_4q_1$ & $s_{11}~ ac_6m_5q_2r_4y_4$ & \myway{$s_7~ ac_5m_2m_3p_1$} & $s_7~ ac_6m_3p_4$ & $s_6~ ap_3p_5z$ \\ \hline
47 & $s_5~ c_4d_2k$ & $s_9~ ad_3m_4q_1y_3$ & \myway{$s_7~ ac_6m_4m_5p_2$} & $s_7~ ac_5m_2m_3p_1$ & $s_7~ ac_6m_3p_4$ & $s_6~ ap_3p_5z$ \\ \hline
48 & $s_5~ c_4d_2ky_2$ & \myway{$s_7~ ac_3m_4p_1$} & $s_7~ ac_6m_4m_5p_2$ & $s_7~ ac_5m_2m_3p_1$ & $s_7~ ac_6m_3p_4$ & $s_6~ ap_3p_5z$ \\ \hline
49 & \myway{$s_{12}~ ac_2c_4kw^{2}$} & $s_7~ ac_3m_4p_1v$ & $s_7~ ac_6m_4m_5p_2$ & $s_7~ ac_5m_2m_3p_1v$ & $s_7~ ac_6m_3p_4$ & $s_6~ ap_3p_5z$ \\ \hline
50 & $s_{12}~ ac_2c_4kv^{2}w$ & $s_{12}~ ac_3m_4p_1vw^{2}$ & $s_7~ ac_6m_4m_5p_2v^{2}$ & $s_{12}~ ac_5m_2m_3p_1vw^{2}$ & $s_7~ ac_6m_3p_4v$ & $s_6~ ap_3p_5z$ \\ \hline
51 & $s_{12}~ ac_2c_4k$ & $s_{12}~ ac_3n_4p_1vw$ & $s_{12}~ ac_6m_4m_5p_2v^{2}w^{2}$ & $s_{12}~ ac_5n_2n_3p_1v^{2}w$ & $s_{12}~ ac_6m_3p_4vw^{2}$ & $s_6~ ap_3p_5v^{2}z$ \\ \hline
52 & $s_5~ ac_2c_4k$ & $s_{12}~ ac_3n_4p_1$ & $s_{12}~ ac_6n_4n_5p_2vw$ & $s_{12}~ ac_5n_2n_3p_1$ & $s_{12}~ ac_6n_3p_4vw$ & $s_{12}~ ap_3p_5vw^{2}z$ \\ \hline
53 & $s_{13}~ a^{2}c_2c_4w^{2}$ & $s_7~ a^{2}c_3n_4p_1$ & $s_{12}~ ac_6n_4n_5p_2$ & $s_7~ a^{2}c_5n_2n_3p_1$ & $s_{12}~ ac_6n_3p_4$ & $s_{12}~ ap_3p_5wz$ \\ \hline
54 & $s_{13}~ a^{4}c_2c_4w$ & $s_{13}~ a^{3}c_3n_4p_1w^{2}$ & $s_7~ a^{3}c_6n_4n_5p_2$ & $s_{13}~ a^{3}c_5n_2n_3p_1w^{2}$ & $s_7~ a^{2}c_6n_3p_4$ & $s_{12}~ ap_3p_5z$ \\ \hline
55 & $s_{13}~ a^{4}c_2c_4$ & $s_{13}~ a^{4}c_3n_4p_1w$ & $s_{13}~ a^{4}c_6n_4n_5p_2w^{2}$ & $s_{13}~ a^{5}c_5n_2n_3p_1w$ & $s_{13}~ a^{3}c_6n_3p_4w^{2}$ & $s_6~ a^{3}p_3p_5z$ \\ \hline
56 & $s_0~ c_2c_4$ & $s_{13}~ a^{4}c_3n_4p_1$ & $s_{13}~ a^{5}c_6n_4n_5p_2w$ & $s_{13}~ a^{5}c_5n_2n_3p_1$ & $s_{13}~ a^{4}c_6n_3p_4w$ & $s_{13}~ a^{3}p_3p_5w^{2}$ \\ \hline
57 & $s_0~ c_2c_4$ & $s_0~ c_3p_1$ & $s_{13}~ a^{5}c_6n_4n_5p_2$ & $s_0~ c_5p_1$ & $s_{13}~ a^{4}c_6n_3p_4$ & $s_{13}~ a^{3}p_3p_5w$ \\ \hline
58 & $s_0~ c_2c_4$ & $s_0~ c_3p_1$ & $s_0~ c_6p_2$ & $s_0~ c_5p_1$ & $s_0~ c_6p_4$ & $s_{13}~ a^{3}p_3p_5$ \\ \hline
59 & \myway{$s_0~ c_2c_4$} & \myway{$s_0~ c_3p_1$} & \myway{$s_0~ c_6p_2$} & \myway{$s_0~ c_5p_1$} & \myway{$s_0~ c_6p_4$} & \myway{$s_0~ p_3p_5$} \\ \hline
\end{tabular}
\end{center}
\end{table}


\end{document}